\documentclass[11pt]{article}

\usepackage{paper-style}
\theoremstyle{plain}
\newtheorem*{mainrestatement}{%
  \hyperref[thm:main]{Theorem~\ref*{thm:main}}%
}
\usepackage{tikz}
\usepackage{pgfplots,pgfplotstable}
\pgfplotsset{compat=1.17}
\usetikzlibrary{patterns.meta,decorations.markings,arrows.meta}
\usepackage[dvipsnames]{xcolor}
\usepackage{booktabs}
\usepackage{longtable}
\usepackage{amsmath}
\usepackage[most]{tcolorbox}
\usepackage[showdeletions]{color-edits}
\usepackage{etoolbox}
\AtBeginEnvironment{proposition}{\crefalias{lemma}{proposition}}
\AtBeginEnvironment{corollary}{\crefalias{lemma}{corollary}}
\AtBeginEnvironment{fact}{\crefalias{lemma}{fact}}

\def\comments{1}

\ifnum\comments=1
\addauthor{jh}{Sepia}
\newcommand{\pras}[1]{\textcolor{BrickRed}{\sf{[#1 --PR]}}}
\newcommand{\kangning}[1]{\textcolor{RoyalBlue}{\sf{[#1 --KW]}}}
\newcommand{\moses}[1]{\textcolor{ForestGreen}{\sf{[#1 --MC]}}}
\newcommand{\qilin}[1]{\textcolor{Purple}{\sf{[#1 --QY]}}}
\newcommand{\jabari}[1]{\textcolor{Sepia}{\sf{[#1 --JH]}}}
\newcommand{\ziyi}[1]{\textcolor{TealBlue}{\sf{[#1 --ZC]}}}
\newcommand{\fabi}[1]{\textcolor{orange}{\sf{[#1 --FF]}}}

\fi

\ifnum\comments=0
\newcommand{\pras}[1]{}
\newcommand{\kangning}[1]{}
\newcommand{\moses}[1]{}
\newcommand{\qilin}[1]{}
\newcommand{\jabari}[1]{}
\newcommand{\jhcomment}[1]{}
\newcommand{\ziyi}[1]{}
\newcommand{\fabi}[1]{}

\fi

\newcommand{\SC}{\mathsf{cost}}

\DeclareMathOperator{\rank}{rank}

\newcommand{\lowrho}{\underline{\rho}}
\newcommand{\low}[1]{\underline{#1}}
\DeclareMathOperator{\Unif}{Unif}

\title{Stable Voting Rules on the \\Edge of Optimal Metric Distortion}

\author{%
  \begin{tabular}{@{}c@{\hspace{2.5em}}c@{\hspace{2.5em}}c@{}}
    Ziyi Cai & Fabian Frank & Jabari Hastings \\[0.15em]
    {\normalsize Rutgers University} &
    {\normalsize Technical University of Munich} &
    {\normalsize Stanford University} \\[-0.1em]
    {\small\texttt{zc417@cs.rutgers.edu}} &
    {\small\texttt{fabian\_w.frank@tum.de}} &
    {\small\texttt{jabarih@stanford.edu}} \\[1.4em]
    Prasanna Ramakrishnan & Kangning Wang & Qilin Ye \\[0.15em]
    {\normalsize Harvard University} &
    {\normalsize Rutgers University} &
    {\normalsize Stanford University} \\[-0.1em]
    {\small\texttt{pras\_r@seas.harvard.edu}} &
    {\small\texttt{kn.w@rutgers.edu}} &
    {\small\texttt{yql@stanford.edu}}
  \end{tabular}%
}
\begin{document}

\maketitle

\vspace{8pt}

\begin{abstract}
We prove the existence of a randomized voting rule with metric distortion at most $2.1248$, within $0.013$ of the lower bound of $2.11264$. Our rule comes from a generalization of \emph{stable $k$-lotteries} developed in the context of committee selection. 
In contrast to prior work, our rule samples from a single distribution derived from a zero-sum game, without mixing between voting rules. 
Our result also gives sharp distortion bounds for stable $k$-lotteries. In particular we show that stable $2$-lotteries have distortion $7/3$, despite only relying on aggregate preferences over triples of candidates, and that stable $2$-lotteries are optimal in this class of voting rules.

\end{abstract}

\thispagestyle{empty}

\clearpage

\setcounter{tocdepth}{2}
\begingroup
\let\originalcontentsline\contentsline
\renewcommand{\contentsline}[4]{%
  \ifstrequal{#4}{subsection.C.1}{}{%
    \ifstrequal{#4}{subsection.C.2}{}{%
      \originalcontentsline{#1}{#2}{#3}{#4}%
    }%
  }%
}
\tableofcontents
\endgroup

\thispagestyle{empty}
\clearpage

\pagenumbering{arabic}

\section{Introduction}

Elections are a powerful model for how groups of individuals can make collective decisions. The most prominent example is political democracy, but all manner of social choice settings can be viewed as
elections: organizations choosing new hires, conferences choosing papers, judges choosing prizewinners, or even friends deciding where to eat. Voting theory has also recently found new uses in machine learning, where preference aggregation has become an important component of systems that learn from human comparisons (see \citet{halpern2026ai} for a detailed survey).

With these settings in mind, social choice theory aims to develop voting rules that, given voters' preferences, tend to choose the best, most representative candidates. But centuries of research suggest that this is harder than it may seem. Condorcet's 1785 paradox \citep{condorcet1785essai} shows that there are elections in which every candidate loses to another in a majority vote. Arrow's theorem \citep{arrow1950difficulty} and the later Gibbard--Satterthwaite theorem \citep{gibbard1973manipulation,satterthwaite1975strategy} showed that even seemingly modest requirements on voting rules can be impossible to satisfy together. While these influential results tell us what voting rules cannot hope to achieve, they offer less guidance on which rules are better than others, or where to look for promising new ideas.

Looking for a more optimistic perspective, researchers in computational social choice have developed another strategy: studying quantitative objectives instead of the often-unsatisfiable binary properties of classical social choice.
These make it easier to compare voting rules and can inspire new ones by turning the search for better rules into a concrete optimization problem.
Along these lines, the \emph{metric distortion} framework has emerged as a compelling approach \citep{DBLP:conf/aaai/AnshelevichBP15,DBLP:journals/ai/AnshelevichBEPS18}. Inspired by a rich body of work in economics on spatial models of voting \citep{enelow1984spatial,enelow1990advances,merrill1999unified,jessee2012ideology,armstrong2020analyzing}, it imagines that the voters and candidates lie in a shared metric space. Each voter's cost of a candidate is the distance between them, so voters prefer closer candidates. The goal is to choose a candidate whose average cost is as small as possible. Given only each voter's ordinal preferences (ranking the candidates by distance), can we design a voting rule that always chooses a candidate with approximately optimal cost? The worst-case approximation factor for a voting rule is called its \emph{distortion}. (See \Cref{fig:2v-2c} for an illustrative example.)

\begin{figure}[h]
\centering
\normalsize
\begingroup
\renewcommand{\arraystretch}{1.15}

\begin{tabular}{c@{\hspace{2.5em}}c@{\hspace{2.5em}}c}
\textbf{Preferences}
& \textbf{Possible Metrics}
& \textbf{Social Costs}
\\[1em]

\begin{tabular}{l}
$v_1:\ a \succ b$ \\[0.6em]
$v_2:\ b \succ a$
\end{tabular}
&
\begin{tikzpicture}[
  x=5.4cm, y=1cm,
  baseline=(current bounding box.center),
  font=\normalsize
]
\def\r{2.8pt}

\foreach \row/\one/\two in {0/0.5/1, -2.1/0/0.5} {
  \begin{scope}[shift={(0,\row)}]
    \fill[black!6, rounded corners=6pt]
      (-0.09,-0.85) rectangle (1.09,0.85);

    \draw[line width=0.8pt] (0,0) -- (1,0);

    \foreach \x/\name in {0/a, 1/b} {
      \filldraw[fill=red!85!black, draw=black]
        ([yshift=\r]\x,0) circle[radius=\r];
      \node[red!85!black, above=5pt]
        at ([yshift=\r]\x,0) {$\name$};
    }

    \foreach \x/\idx in {\one/1, \two/2} {
      \filldraw[fill=RoyalBlue, draw=black]
        ([yshift=-\r]\x,0) circle[radius=\r];
      \node[RoyalBlue, below=5pt]
        at ([yshift=-\r]\x,0) {$v_{\idx}$};
    }
  \end{scope}
}
\end{tikzpicture}
&
\begin{tikzpicture}[
  y=1cm,
  baseline=(current bounding box.center),
  font=\normalsize
]
\path[use as bounding box] (0,-2.95) rectangle (0,0.85);
\node at (0,0)
  {$\SC(a)=3\,\SC(b)$};
\node at (0,-2.1)
  {$\SC(b)=3\,\SC(a)$};
\end{tikzpicture}
\end{tabular}
\endgroup
\caption{A simple election with two candidates $a$ and $b$ and two disagreeing voters $v_1$ and $v_2$. The two illustrated metrics show that all deterministic voting rules have distortion at least $3$ and all randomized rules have distortion at least $2$.}\label{fig:2v-2c}
\end{figure}

While spatial models of voting are well studied, the metric assumption may still feel like a curious choice. A few points help explain why this framework has led to a fruitful line of inquiry. 
First, because the metric space is unknown, a low-distortion voting rule must make a choice that is approximately optimal simultaneously in \emph{every} metric space consistent with the voters' preferences. Thus, the point is less that a particular metric space models reality, and more that such a strong worst-case guarantee suggests the voting rule does not behave unreasonably.\footnote{As another indicator of robustness, \cite{DBLP:conf/netecon/GoelHK18} showed that ``average cost'' distortion $D$ automatically implies distortion $D + 2$ with any monotone symmetric norm of the costs.} It may even seem too good to be true --- how can we always find an approximately optimal candidate with so little information about the distances? But it turns out that the structure imposed by the triangle inequality is just enough for several voting rules to have \emph{constant} distortion. This contrasts sharply with the closely related \emph{utilitarian} distortion framework of \cite{DBLP:conf/cia/ProcacciaR06}, in which distortion is unbounded in general, and grows with the number of candidates even under normalization assumptions. Finally, although many rules have small distortion, all well-studied voting rules fall short of the \emph{optimal} distortion. Successive improvements in the best known bounds have required new voting rules, often revealing simple and natural ideas that might otherwise have been missed.

\paragraph{The pursuit of optimal distortion.} In their work introducing the metric distortion problem, \citet{DBLP:conf/aaai/AnshelevichBP15,DBLP:journals/ai/AnshelevichBEPS18} used a simple election with two candidates and two disagreeing voters to show that all deterministic voting rules have distortion at least $3$ (\Cref{fig:2v-2c}). They conjectured that this lower bound is tight, but only proved an upper bound of $5$, attained by the \emph{Copeland rule}.\footnote{Fun fact: The Copeland rule is one of the oldest written voting rules, originating in the work of 13th century philosopher Ramon Llull \citep{hagele2001llull}. It chooses the candidate that beats the most other candidates in a majority vote. Surprisingly, it has the best distortion among deterministic voting rules developed before the metric distortion framework.} A long line of work \citep{DBLP:conf/sigecom/GoelKM17,DBLP:conf/ec/MunagalaW19,DBLP:conf/aaai/Kempe20a} developed new voting rules with stronger distortion guarantees, culminating in a proof of the conjecture \citep*{DBLP:conf/focs/GkatzelisHS20}. Later, \cite{DBLP:conf/ijcai/KizilkayaK22,DBLP:conf/sigecom/Kizilkaya023} introduced intuitive voting rules, starting with \emph{Plurality Veto}, that gave a simple and elegant proof of the same result. At a high level, these rules assign each candidate a score, initially their number of first choice votes (plurality score). Voters then reduce (veto) the score of their least favorite candidate whose score is still positive. 

For randomized rules, the story is more complicated. \cite{DBLP:journals/jair/AnshelevichP17,DBLP:conf/sigecom/FeldmanFG16} showed that distortion 3 is easy with randomization: the \emph{Random Dictatorship} rule simply outputs a random voter's favorite candidate. On the other hand, they showed that the same simple example in \Cref{fig:2v-2c} implies that any randomized rule has distortion at least $2$, and \cite{DBLP:conf/sigecom/GoelKM17} conjectured that this bound is tight. This conjecture was disproved independently by \cite{DBLP:conf/soda/CharikarR22,DBLP:journals/corr/abs-2111-08698}, raising the lower bound to 2.11264. Improving the upper bound of 3 proved challenging, in part because of barriers ruling out improvements for wide classes of voting rules.\footnote{Distortion 3 is also a natural threshold because it is the metric distortion of a \emph{Condorcet winner}, a candidate that beats all others in a majority vote. While a Condorcet winner does not always exist, a distortion 3 candidate does.} In addition to deterministic rules, these include rules that only consider aggregate pairwise comparisons between candidates\footnote{These are called \emph{tournament rules}; examples include Copeland, Borda, and maximal lotteries.} \citep{DBLP:conf/sigecom/GoelKM17}, and rules that always choose some voter's top choice\footnote{Examples include Plurality Voting, Instant Runoff Voting, Random Dictatorship, and (Simultaneous) Plurality Veto.} \citep{DBLP:conf/aaai/GrossAX17}.

Finally, \cite*{DBLP:journals/jacm/CharikarRWW24} devised a voting rule with distortion at most 2.75271, proving the first separation between the metric distortion of deterministic and randomized voting rules. Their voting rule has two components: \emph{maximal lotteries}, a randomized voting rule due to \cite{kreweras1965aggregation} derived from the equilibrium of a simple zero-sum game, and a new voting rule called \emph{RaDiUS}, which chooses a random voter's favorite candidate within the \emph{weighted uncovered set} (introduced by \cite{DBLP:conf/ec/MunagalaW19}). Their proof starts by showing that maximal lotteries have distortion 3, in a way that characterizes which metrics come close to the worst case. These worst-case instances for maximal lotteries turn out to be best-case instances for RaDiUS, so a carefully chosen mixture of the two can achieve significantly lower distortion than either rule alone.

Despite significant attention, this bound stood for more than three years, until two independent works by \cite{frank2026improved,ye2026half}. They follow the prior template of mixing maximal lotteries with a voting rule with similarities to Random Dictatorship. This second rule is a clever new variant of the \emph{Simultaneous Plurality Veto} rule of \citet{DBLP:conf/sigecom/Kizilkaya023} (called Integrated Veto by \cite{frank2026improved} and Veto Lottery by \cite{ye2026half}) which chooses a candidate with probability proportional to their average score over the veto process. They give remarkably simple proofs that an \emph{equal mixture} of this rule and maximal lotteries has distortion 2.5. While these results nearly halve the gap, closing it remains an enticing open problem.

\subsection{Our contributions and technical overview}

In this paper, we give a randomized voting rule with nearly optimal distortion, bringing the upper bound within a hair of the lower bound of $2.11264$.

\begin{theorem}\label{thm:main}
There is a randomized voting rule with distortion at most $2.1248$.
\end{theorem}

We begin with a brief overview of the voting rules we study and the ideas behind their analysis. In \Cref{sec:voting-rules}, we develop the rules step by step, starting with maximal lotteries and building up to the $g$-stable lotteries used for our main result. Along the way, we explain the ideas behind these rules and introduce the properties that will drive the distortion analysis.

\paragraph{The voting rules.} Our starting point is a family of voting rules derived from the equilibria of games.
For maximal lotteries, the game is particularly simple: two players each choose a candidate, and a uniformly random voter decides the winner by choosing their favorite of the two. An equilibrium strategy is a distribution over candidates, which we can use as a randomized voting rule. 

\citet*{DBLP:journals/teco/ChengJMW20} generalized this idea to \emph{stable lotteries}, to study proportional representation in \emph{committee selection}, where the goal is to choose $k$ winners rather than one. The defender draws $k$ candidates independently from a common distribution, the attacker chooses one, and a random voter picks their favorite among all $k+1$. The defender's equilibrium distribution is a \emph{stable $k$-lottery}; to select a single winner, we simply draw once from it. 

Even this family gives surprisingly strong distortion guarantees. We prove a tight general bound for stable $k$-lotteries in \Cref{sec:stable-k}, with distortion at most $7/3\approx 2.333$ for both $k=2$ and $k=3$ (\Cref{thm:dist-sl}, with matching lower bound in \Cref{thm:lower-bound-stable-k}). A remarkable aspect of this result is how little preference information it requires. Motivated in part by common preference elicitation methods in machine learning, \citet*{DBLP:conf/sigecom/CharikarRT025} studied \emph{$k$-tournament rules}, which only see the aggregate rankings of tuples of $k$ candidates. They showed that a $3$-tournament rule combining stable lotteries with pruning and maximal lotteries can achieve distortion strictly below $3$ (by a tiny, unspecified amount). But a stable $2$-lottery is already a $3$-tournament rule: its game only needs to know how often each candidate is the favorite within a triple. Our analysis shows that this rule alone achieves the substantially stronger bound of $7/3$. In \Cref{app:small-tournament} we show that stable $2$-lotteries are optimal among $3$-tournament rules.

While studying \emph{Condorcet winning sets} in committee selection, \citet{DBLP:conf/stoc/CharikarLRV025} introduced \emph{activation functions} to generalize stable $k$-lotteries to what we call \emph{$g$-stable lotteries}. If a voter prefers a challenger to an $r$-fraction of a distribution, it beats $k$ independent draws with probability $r^k$. We can replace this power by \emph{any} continuous increasing function $g(r)$: Kakutani's fixed-point theorem (also used to establish Nash equilibria) guarantees that a $g$-stable lottery exists. This opens up a vast design space, and once we can analyze the distortion of a $g$-stable lottery, it only remains to find the choice of $g$ that gives the best distortion guarantee.

Considering that this family of voting rules was developed outside metric distortion, for questions about representation and stability,
it is striking that they can come so close to optimal distortion. They are also substantially simpler conceptually than the rules behind previous improvements: we just sample from one cleanly defined distribution, without having to work out which rules compensate for each other's weaknesses. This answers an explicit question of \citet{DBLP:journals/jacm/CharikarRWW24}, who asked for a simpler rule that breaks the barrier of $3$ without mixing voting rules.

\paragraph{The analysis.} The main ideas in our analysis came from GPT-5.6 Sol, and later GPT-6 Astra (see our AI methodology in \Cref{sec:discussion}). The starting point is the \emph{biased metric} framework of \citet{DBLP:conf/soda/CharikarR22}, which we review in \Cref{sec:biased-metrics}. For any fixed election, this framework reduces the worst-case metrics to a family described by just $m$ parameters, one for each candidate. These parameters, expressed using the function $\rho:C\to\mathbb{R}_{\geq 0}$, come from the candidates' distances to the optimal candidate $a^\star$, and determine the relevant voter-candidate distances through the rankings.

Much of the prior work analyzes distortion by tracking the fraction of voters with particular preferences. We instead try to bound each voter's contribution separately, and only average over voters at the end. At first, this seems too much to ask: a voter might be arbitrarily close to $a^\star$ and have large expected cost for the lottery. The workaround is to use a carefully chosen potential function $\Phi$ to
correct each voter's contribution. This gives us room to make the individual inequalities hold, provided the corrections cancel when we average over voters. The \emph{stationarity} of $g$-stable lotteries is the magical ingredient that allows this cancellation to occur.

For example, a stable $k$-lottery satisfies the stationary property that if we sample a random voter $v$ and let $\hat b_v$ be
voter $v$'s favorite among $k+1$ independent draws from $D$, then $\hat b_v$ has the same distribution as $D$ (\Cref{prop:sl-stationary}). But for an individual voter $v$, $\hat b_v$ tends to be a lot better than a single draw $b\sim D$, and this difference is most noticeable for the voters whose costs are hardest to bound (like those close to $a^\star$). By offsetting the inequality we hope to prove by the expected difference in potentials $\Phi(\rho(\hat b_v))-\Phi(\rho(b))$, we can actually make the inequality true for every voter while having the offsets cancel by stationarity when averaging over voters.

In \Cref{sec:optimal-g}, we extend this approach to obtain a nearly tight characterization of the optimal distortion for $g$-stable lotteries. We begin by constructing three simple families of elections, accompanied by ``simple'' metrics that take only a handful of integer values (in the spirit of the lower bound of \cite{DBLP:conf/soda/CharikarR22}). The constraints imposed by the first two families show that no $g$-stable lottery can achieve distortion better than $2.1232$, while combining them in the third family guides our choice of a (near-)optimal activation function $g^\star$. To establish our upper bound of $2.1248$ for $g^\star$, we refine the voter-by-voter analysis developed for stable $k$-lotteries, introducing auxiliary functions to more carefully account for the slack in each voter's inequality, and using computer verification to certify that our chosen functions satisfy the analytic conditions we impose.

\subsection{Concurrent work}

Independently and concurrently with our work, a recent preprint of \cite{shah2026improving} proves a distortion upper bound of roughly $2.136$. His voting rule, \emph{random-size stable lotteries}, is a variant of stable lotteries where the integer number of draws ($k \ge 1$) is randomized. This variant is conceptually similar to $g$-stable lotteries, corresponding to the special case in which $g$ is \emph{absolutely monotone} on $[0,1]$, meaning its derivatives of all orders are nonnegative. In particular, if $p_k$ is the probability of taking $k$ draws, the corresponding activation function is $g(r) = \sum_{k\geq 1} p_k r^k$. \cite{shah2026improving} also proves a lower bound within $10^{-4}$ for the family of random-size stable lotteries; together with our upper bound of $2.1248$, this establishes a distortion separation between optimal $g$-stable lotteries and optimal random-size stable lotteries.

\section{Preliminaries and notation}\label{sec:prelims}

\paragraph{Elections} An \emph{election} (or \emph{preference profile}) is a tuple $\mathcal{E}=(V,C,\{\succ_v\}_{v\in V})$, where $V$ is a set
of $n$ voters, $C$ is a finite set of candidates, and for each $v\in V$, $\succ_v$ is a total linear order over the candidates. We write $a \succ_v b$ to denote that voter $v$ prefers candidate $a$ over candidate $b$. We generally use $v$ to denote a voter, $a, b, c$ to denote candidates (in particular,  $a^\star$ for the ``optimal'' candidate, and $b$ or some variant thereof for a chosen or sampled candidate, and $c$ for an auxiliary candidate).

\paragraph{Distributions, lotteries, and ranks.} We use $v\sim V$ to denote sampling a voter uniformly at random, and use the typical notation $\Unif(r, s)$ to denote the uniform distribution over $[r, s]$. A \emph{lottery} is a distribution $D$ over candidates and $D(a)$ denotes the probability mass of a candidate $a$ in $D$. Samples displayed together are always treated as independent. For convenience we will adopt the convention that each candidate corresponds to a continuum of copies, with ties broken uniformly, and sampling a candidate from a lottery means sampling a uniformly random copy. This convention streamlines the statements of theorems like \Cref{thm:ml-key-property,thm:sl-key-property,thm:gsl-key-property} (and the surrounding prose), without needing special cases to deal with ties.

We describe a helpful visualization of a voter's preferences with respect to a lottery that we will use repeatedly.\footnote{This visualization originates in \cite{DBLP:conf/stoc/CharikarLRV025}.} Imagine that we create an interval for each candidate $a \in C$ with width $D(a)$, and arrange these intervals on the interval $[0, 1]$ in increasing order of $v$'s preference. The result is what we call the \emph{rank line} of $v$ with respect to $D$, depicted in \Cref{fig:rank-line}.

\begin{figure}[ht]
\centering
\begin{tikzpicture}[x=12cm,y=1cm,font=\normalsize]
  \definecolor{rankPink}{RGB}{203,125,177}
  \definecolor{rankSky}{RGB}{86,180,233}
  \definecolor{rankGreen}{RGB}{0,158,115}
  \definecolor{rankGold}{RGB}{230,159,0}
  \definecolor{rankBlue}{RGB}{0,114,178}
  \definecolor{rankOrange}{RGB}{213,94,0}
  \definecolor{rankRose}{RGB}{204,121,167}
  \definecolor{rankSoftBlue}{RGB}{112,160,205}
  \definecolor{rankSoftGreen}{RGB}{95,185,145}
  \definecolor{rankPurple}{RGB}{145,120,210}
  \definecolor{rankYellow}{RGB}{248,230,55}

  \draw[rankPink,line width=4pt]      (0.00,0) -- (0.04,0);
  \draw[rankSky,line width=4pt]       (0.04,0) -- (0.21,0);
  \draw[rankGreen,line width=4pt]     (0.21,0) -- (0.27,0);
  \draw[rankGold,line width=4pt]      (0.27,0) -- (0.40,0);
  \draw[rankBlue,line width=4pt]      (0.40,0) -- (0.45,0);
  \draw[rankOrange,line width=4pt]    (0.45,0) -- (0.63,0);
  \draw[rankRose,line width=4pt]      (0.63,0) -- (0.70,0);
  \draw[rankSoftBlue,line width=4pt]  (0.70,0) -- (0.82,0); %
  \draw[rankSoftGreen,line width=4pt] (0.82,0) -- (0.86,0); %
  \draw[rankPurple,line width=4pt]    (0.86,0) -- (1.00,0); %

  \draw[black,line width=1.2pt] (0,-0.13) -- (0,0.13);
  \draw[black,line width=1.2pt] (1,-0.13) -- (1,0.13);
  \node[anchor=east,xshift=-7pt] at (0,0) {worst};
  \node[anchor=west,xshift=7pt] at (1,0) {best};

  \def\labelY{-0.42}
  \node[anchor=base] at (0.000,\labelY) {$0$};
  \node[anchor=base] at (0.665,\labelY) {$\ldots$};
  \node[anchor=base] at (0.760,\labelY) {$c$};
  \node[anchor=base] at (0.840,\labelY) {$b$};
  \node[anchor=base] at (0.930,\labelY) {$a$};
  \node[anchor=base] at (1.000,\labelY) {$1$};

  \filldraw[fill=rankYellow,draw=black!45,line width=0.3pt]
    (0.73,0) circle[radius=0.13cm];

  \draw[rankPurple,line width=2pt] (0.86,-0.13) -- (0.86,0.13);
  \node[rankPurple,anchor=south] at (0.86,0.22) {$\rank_v(a;D)$};
\end{tikzpicture}
\caption{The rank line for a voter $v$ with preference $a\succ b \succ c \succ \cdots$ with respect to a lottery $D$. Each colored block corresponds to a candidate with width equal to their probability mass in $D$, and they are arranged in increasing order of $v$'s preference from left to right.}\label{fig:rank-line}
\end{figure}

We define the \emph{rank} of a candidate $a$ with respect to distribution $D$ as:
$$\rank_v(a;D) = \Pr_{b\sim D}[a\succ_v b].$$
For intuition, higher rank is better (more preferred).
Visually, the rank of a candidate is the coordinate of the left boundary of $a$'s interval. The rank line is convenient because it gives us an elegant way of viewing samples from $D$ from the perspective of a voter: $b\sim D$ is equivalent to drawing $r \sim \Unif(0, 1)$, and taking the candidate whose interval occupies rank $r$. With the convention that each named candidate corresponds to a continuum of copies, it is natural to think of each point on a candidate's interval as corresponding to a copy.

\paragraph{Metric spaces, social cost, and distortion.} We will treat the set of voters and candidates $V\cup C$ as occupying points in a pseudometric space with distance $d(\cdot, \cdot)$ satisfying the following three properties.
\begin{enumerate}[label=(\arabic*)]
    \item Non-negativity: $d(x, y) \geq 0$ (with equality if $x = y$).
    \item Symmetry: $d(x, y) = d(y, x)$.
    \item Triangle inequality: $d(x, y) + d(y, z) \geq d(x, z)$.
\end{enumerate}

We say that $d$ is \emph{consistent} with the election if $a\succ_v b$ implies that  $d(v,a)\leq d(v,b)$ for every voter $v$ and candidates $a,b$. The social cost of a candidate and
of a lottery are denoted 
\[
    \SC_d(a):=\Ev_{v\sim V}[d(v,a)]
    \qquad\text{and}\qquad
    \SC_d(D):=\Ev_{a\sim D}[\SC_d(a)].
\]
We omit the subscript $d$ when the metric is fixed, and write
$a^\star\in\arg\min_{a\in C}\SC_d(a)$ for an optimal candidate. The
distortion of a randomized voting rule $f$ is
\begin{equation*}
    \operatorname{distortion}(f)
    :=\sup_{\mathcal{E}}
      \sup_{d\text{ consistent with }\mathcal{E}}
      \frac{\SC_d(f(\mathcal{E}))}
           {\min_{a\in C}\SC_d(a)}.
\end{equation*}

\section{Voting rules in equilibrium}\label{sec:voting-rules}

Our goal in this section is to build up to $g$-stable lotteries, the family of voting rules behind \Cref{thm:main}. We start with maximal lotteries and introduce the generalizations to stable $k$-lotteries and $g$-stable lotteries one at a time, developing intuition for the rules along with the properties we will use to analyze their distortion.

When designing a randomized voting rule, a direct approach feels most natural. We can define some explicit randomized process, in the way that Random Dictatorship samples a random voter and uses their preferences. Or, we can choose a way to assign ``scores'' to candidates, and define a distribution with candidate probabilities proportional to their scores. For example, Borda scores or plurality scores lead to Proportional Borda and Random Dictatorship again. 

But there is a different, more implicit approach which has been at the heart of several recent results in computational social choice (including on the distortion problem).\footnote{See \citet*{charikar2026game} for a recent survey on this line of work.} We can first devise a two-player game in which each player's pure strategies correspond to candidates, and then a mixed-strategy equilibrium is a distribution over candidates (lotteries) which we can use as a randomized voting rule. For example, consider the following game between two players called the \emph{attacker} and the \emph{defender}.

\definecolor{NiceBoxTitle}{HTML}{EDE5F7}
\definecolor{NiceBoxFrame}{HTML}{9278AD}

\newtcolorbox{NiceBox}[2][]{
  enhanced,
  breakable,
  colback=white,
  colframe=NiceBoxFrame,
  colbacktitle=NiceBoxTitle,
  coltitle=black,
  title={#2},
  fonttitle=\normalfont\bfseries,
  halign title=center,
  boxrule=0.55pt,
  arc=1.6mm,
  outer arc=1.6mm,
  left=8pt,
  right=8pt,
  top=7pt,
  bottom=7pt,
  toptitle=6pt,
  bottomtitle=6pt,
  before skip=0.85\baselineskip,
  after skip=0.85\baselineskip,
  #1
}

\begin{NiceBox}{The Condorcet Game}
  \begin{itemize}
    \item The attacker and defender choose lotteries $D^\dagger$ and $D$.

    \item Sample $a \sim D^\dagger$, $b \sim D$, and a voter $v$ uniformly at random.

    \item The attacker wins if $a\succ_v b$, and the defender wins if $b\succ_v a$. %

  \end{itemize}
\end{NiceBox}

By symmetry, both players have the same equilibrium distribution, called a \emph{maximal lottery}. Maximal lotteries were introduced by \citet{kreweras1965aggregation}, and first studied in depth by \citet{fishburn1984probabilistic}.

\begin{definition}
Given an election, a \emph{maximal lottery} is an equilibrium strategy of the Condorcet game, given by a solution to the following minimax optimization:
\begin{equation}\label{eq:minimax-ml}
\arg\min_{D} \max_{D^\dagger} \Pr_{\substack{v\sim V\\a \sim D^\dagger\\ b\sim D}}[a\succ_v b].
\end{equation}
\end{definition}

Maximal lotteries have several useful properties, and we take a moment to discuss two that are particularly relevant to this paper.

As an easy consequence of the minimax theorem of \citet{v1928theorie}, the value of the optimization in \Cref{eq:minimax-ml} is $\frac12$, meaning that a player sampling from a maximal lottery therefore wins with probability at least $\frac12$, no matter how the other player chooses their candidate.

\begin{theorem}\label{thm:ml-key-property}
In any election with maximal lottery $D$, we have
\begin{equation}\label{eq:ml-key-property}
\Pr_{\substack{v\sim V\\b\sim D}}[a\succ_v b] \leq \frac12
\end{equation}
for all candidates $a$.
\end{theorem}

\Cref{thm:ml-key-property} makes a maximal lottery a randomized analogue of a \emph{Condorcet winner}, a candidate that beats every other candidate in a majority vote.  Here, a random candidate drawn from a maximal lottery beats each candidate for a majority of voters in expectation. 

Just as a Condorcet winner guarantees distortion 3, \citet*{DBLP:journals/jacm/CharikarRWW24} showed that a maximal lottery has distortion 3. They also showed that several voting rules complement maximal lotteries in such a way that mixing between them guarantees distortion strictly less than 3. In prior work, every voting rule with distortion less than 3 follows this template, mixing between maximal lotteries (or an approximate maximal lottery, in the case of \citet*{DBLP:journals/corr/abs-2602-08871}) and another voting rule \citep{DBLP:journals/jacm/CharikarRWW24,DBLP:conf/sigecom/CharikarRT025,DBLP:journals/corr/abs-2602-08871,frank2026improved,ye2026half}.

To see how \Cref{thm:ml-key-property} might generalize, it helps to express it in terms of ranks. 
For the analysis, it is convenient to slightly change what we mean by the rank of an \textit{attacker} $a$. Instead of using the left boundary, we use a uniformly random point in $a$'s interval, corresponding to a random copy of $a$. Overloading the notation, we still write $\mathrm{rank}_v(a;D)$ for the attacker's rank.
Then, we have
$$\Pr_{\substack{v\sim V\\b\sim D}}[a\succ_v b] =  \Ev_{v\sim V}[\Pr_{b\sim D}[a\succ_v b]] = \Ev_{v\sim V}[\rank_v(a; D)].$$
Thus, \Cref{eq:ml-key-property} is equivalent to
\begin{equation}\label{eq:ml-key-property-rank}
\Ev_{v\sim V}[\rank_v(a; D)] \leq \frac12.
\end{equation}

Maximal lotteries also satisfy a \emph{stationarity} property, whose generalizations will be central to our distortion analysis.\footnote{This property is observed explicitly by \citet{brandl2024natural} for maximal lotteries, and earlier by \citet{laslier2017reinforcement} for a closely related tournament game. Related dynamics based on pairwise comparisons were studied earlier by \citet{ganikhodzhaev1993quadratic}.} 
We state it here and prove it in the general $g$-stable setting in \Cref{prop:gsl-stationary}.

\begin{proposition}\label{prop:ml-stationary}
Given a maximal lottery $D$, define a distribution $\widehat{D}$ as follows. Draw two candidates from $D$ and output the favorite candidate of a voter $v\sim V$ sampled uniformly at random. Then $\widehat{D} = D$. 
\end{proposition}

There are many ways of modifying the Condorcet game to derive relatives of maximal lotteries. One such relative is the \emph{stable lottery}, originally introduced by \citet*{DBLP:journals/teco/ChengJMW20} in the context of \emph{committee selection}.\footnote{For clarity, we note two minor differences between the original formulation and what we use here. First, \citet{DBLP:journals/teco/ChengJMW20} considered distributions over committees of size $k$, and we instead use an independent sampling version made explicit by \cite{DBLP:conf/sigecom/CharikarRT025}. Second, \citet{DBLP:journals/teco/ChengJMW20} actually defined stable lotteries as distributions satisfying \Cref{thm:sl-key-property} (in fact, with $\frac{1}{k}$ to align with usual notions of stability in place of the $\frac{1}{k + 1}$ implied by the proof) rather than directly defining them in terms of the minimax problem.} In this setting, the goal is to choose a set of $k$ winners rather than a single winner, and so one might naturally consider a version of the Condorcet game in which one player chooses $k$ candidates.\footnote{Stable lotteries can also be defined by considering a symmetric $k + 1$ player game in which each player chooses a distribution over candidates, and a uniformly random voter picks the winner from the samples. Recently, \cite{CaiZheng2026Asymptotic} used this approach in the context of AI alignment.}

\begin{NiceBox}{The $k$ vs $1$ Condorcet Game}
  \begin{itemize}
    \item The attacker and defender choose lotteries $D^\dagger$ and $D$.

    \item Sample $a \sim D^\dagger$, $b_1, \dots, b_k \sim D$, and a voter $v$ uniformly at random.

    \item The attacker wins if $a\succ_v b_1,\dots, b_k$, and the defender wins otherwise. %

  \end{itemize}
\end{NiceBox}

\begin{definition}
Given an election, a \emph{stable $k$-lottery} is an equilibrium distribution for the defender in the $k$ vs $1$ Condorcet game, given by a solution to the following minimax optimization:
\begin{equation}\label{eq:minimax-sl}
\arg\min_{D} \max_{D^\dagger} \Pr_{\substack{v\sim V\\a \sim D^\dagger\\ b_1,\dots,b_k\sim D}}[a\succ_v b_1,\dots,b_k].
\end{equation}
\end{definition}

Although this game is no longer symmetric, the defender's loss is convex in $D$ and linear in $D^\dagger$, so the minimax theorem still applies. The defender's optimal value is unchanged if they choose after seeing the attacker's lottery, swapping the $\min$ and $\max$ in \Cref{eq:minimax-sl}. But then the defender can simply copy that lottery. Each of the $k+1$ independent samples $\{a,b_1,\dots,b_k\}$ is equally likely to be the random voter's favorite, so the defender loses with probability at most $\frac{1}{k+1}$. This gives the following guarantee.%

\begin{theorem}\label{thm:sl-key-property}
In any election with stable $k$-lottery $D$, we have
\begin{equation}\label{eq:sl-key-property}
\Pr_{\substack{v\sim V\\ b_1,\dots,b_k\sim D}}[a\succ_v b_1,\dots,b_k] \leq \frac{1}{k + 1}.
\end{equation}
for all candidates $a$.
\end{theorem}

Stationarity also extends just as naturally. We now draw $k+1$ candidates and let a random voter choose their favorite.

\begin{proposition}\label{prop:sl-stationary}
Given a stable $k$-lottery $D$, define a distribution $\widehat{D}$ as follows. Draw $k + 1$ candidates from $D$ and output the favorite sample of a voter $v\sim V$ sampled uniformly at random. Then $\widehat{D} = D$. 
\end{proposition}

We can also reinterpret \Cref{eq:sl-key-property} in terms of the ranks. For a fixed voter and a fixed copy of attacker $a$, the attacker beats all $k$ independent draws with probability $\rank_v(a;D)^k$. Averaging over voters and the attacker copy gives

$$\Pr_{\substack{v\sim V\\ b_1,\dots,b_k\sim D}}[a\succ_v b_1,\dots,b_k] = \Ev_{b_1,\dots,b_k \sim D}[\Pr_{v\sim V}[a\succ_v b_1, \dots, b_k]] $$
$$= \Ev_{v\sim V}[\Pr_{b \sim D}[a\succ_v b]^k] = \Ev_{v\sim V}[\rank_v(a;D)^k],$$

Thus, \Cref{eq:sl-key-property} is equivalent to
\begin{equation}\label{eq:sl-key-property-rank}
\Ev_{v\sim V}[\rank_v(a;D)^k] \leq \frac{1}{k + 1}
\end{equation}
and the objective in \Cref{eq:minimax-sl} can be written in terms of ranks as
\begin{equation}\label{eq:minimax-sl-ranks}
\min_{D} \max_{D^\dagger} \Ev_{\substack{v\sim V\\a\sim D^\dagger}}[\rank_v(a;D)^k].    
\end{equation}

The rank formulation suggests another question: why should the payoff be the particular power $r^k$? We can replace it by a function $g(r)$ and instead consider
\begin{equation}\label{eq:minimax-gsl-ranks}
\min_{D} \max_{D^\dagger} \Ev_{\substack{v\sim V\\a\sim D^\dagger}}[g(\rank_v(a;D))].   
\end{equation}
For $g(r)=r^k$, we recover the stable $k$-lottery objective. Which other functions give an equilibrium with analogous properties?

\citet*{DBLP:conf/stoc/CharikarLRV025} used the minimax theorem to show that this works when $g$ is increasing and convex. They called $g$ the \emph{activation function}, since it ``activates'' the parts of the rank interval that should be prioritized.\footnote{Also, the optimal $g$ in their setting happened to resemble a ReLU function after transformation.} Follow-up work by \cite*{DBLP:conf/soda/SongNL26} could be understood as using an activation function which is non-convex (in particular, a simple threshold function), and \cite*{nguyen2026stable,charikar2026exposition} make this generalization formal. 

\begin{definition}
Fix a continuous, non-decreasing function $g: [0, 1]\to [0, 1]$. Given an election, a \emph{$g$-stable lottery}, denoted $D^g$, is a solution to the following minimax optimization:
\begin{equation}\label{eq:minimax-gsl}
\arg\min_{D} \max_{D^\dagger} \Ev_{\substack{v\sim V\\a\sim D^\dagger}}[g(\rank_v(a;D))].
\end{equation}
\end{definition}

Restricting the range of $g$ to $[0,1]$ is unnecessary, but a convenient normalization which allows us to interpret $g(r)$ as a probability. The existence guarantee also does not require $g$ to be non-decreasing, but this requirement is natural to place more value on higher ranks, and is needed for the key properties \Cref{thm:gsl-key-property} and \Cref{prop:gsl-stationary}.
For the curious reader, we can also think of $D^g$ as the minimax-optimal strategy\footnote{This strategy may not be part of an equilibrium if $g$ is nonconvex, but the minimax solution still satisfies the guarantees we need.}
for the defender in the following zero-sum game. 

\begin{NiceBox}{The $g$-Activated Condorcet Game}
  \begin{itemize}
    \item The attacker and defender choose lotteries $D^\dagger$ and $D$.

    \item Sample $a\sim D^\dagger$ and a voter $v$ uniformly at random.

    \item Let $\displaystyle r=\rank_v(a;D)=\Pr_{b\sim D}[a\succ_v b].$ 

    \item The attacker wins with probability $g(r)$, and the defender wins otherwise.
  \end{itemize}
\end{NiceBox}

To give some intuition, the function $g$ tells us how much attention the rule pays to different parts of a voter's ranking (measured in terms of the rank line from \Cref{fig:rank-line}). 
A function $g$ that rises sharply close to $1$ (like $g(r) = r^k$ as $k\to\infty$, corresponding to Random Dictatorship) is primarily 
concerned when a candidate beats almost the whole lottery, which could be natural if we believe that voters' top choices matter the most and the rest is just noise. A milder $g$ (like $g(r) = r$, corresponding to maximal lotteries) also responds to candidates that do well at more modest ranks, like a Condorcet winner that lurks deep in voters' preferences. In general, we can carefully tailor $g$ to bake the priorities of a given setting into our voting rule.

Generalizing \Cref{thm:sl-key-property} to $g$-stable lotteries we get the following theorem. 

\begin{theorem}\label{thm:gsl-key-property}
In any election with $g$-stable lottery $D^g$, we have
\begin{equation}\label{eq:gsl-key-property}
\Ev_{v\sim V}[g(\rank_v(a;D^g))] \leq \int_0^1 g(t) \dif t.
\end{equation}
for all candidates $a$.
\end{theorem}

\begin{proof}
Let 
\[H(D) = \max_{a\in C}  \Ev_{v\sim V} [g(\mathrm{rank} _v(a;D))].\] 
Since the objective defining $D^g$ is linear in $D^{\dagger}$, its maximum over $D^{\dagger}$ equals $H(D)$. Thus, $D^g$ minimizes $H(D)$.

For each $D$, consider the set of distributions supported on the candidates attaining $H(D)$. Since $C$ is finite and $g$ is continuous, this correspondence has nonempty, compact, convex values, as well as a closed graph. Hence, Kakutani's theorem gives a fixed point $D_0$. Thus, every candidate in $\mathrm{supp}(D_0) $ attains $H(D_0)$, and hence
\[
		H(D_0) = \Ev_{\substack{ v \sim  V \\ a \sim  D_0 } } [g(\mathrm{rank} _v (a;D_0))].
\]

For each voter, drawing $a \sim  D_0$ gives a uniformly random rank in $[0,1]$. Therefore, $H(D_0) = \int_{0}^{1} g(t) \;\mathrm{d}t$. Since $D^{g}$ minimizes $H$, every candidate $a$ satisfies
\[
		\Ev_{v\sim V} [g(\mathrm{rank}_v (a;D^{g}) )] \leq H(D^{g}) \leq H(D_0) = \int_{0}^{1} g(t) \;\mathrm{d}t. \qedhere
\]  
\end{proof}

Finally, stationarity has a clean interpretation in terms of a voter's rank line.

\begin{proposition}
\label{prop:gsl-stationary}
Given a $g$-stable lottery $D^g$, define a distribution $\widehat{D}$ as follows. Sample a voter $v\sim V$ uniformly and a rank $r \sim [0, 1]$ with density proportional to $g(r)$, and output the candidate occupying rank $r$. Then $\widehat{D} = D^g$. 
\end{proposition}

For $g(r)=r^k$, sampling a rank with density proportional to $g$ is equivalent to taking the largest of $k+1$ independent uniform ranks. This recovers the earlier experiment of asking a random voter to choose their favorite among $k+1$ draws.

\begin{proof}

Fix a candidate $a$. For a voter $v$, the copies of $a$ occupy the interval  $[\mathrm{rank} _v(a; D^{g}), \allowbreak \mathrm{rank} _v (a; D^{g}) + D^{g}(a))$. As the rule returns $a$ iff the rank we sample lies in this interval, averaging gives
\[
		\left( \int_{0}^{1} g(t) \;\mathrm{d}t \right) \widehat{D}(a) = \Ev_{v \sim  V} \left[  \int_{\mathrm{rank} _v(a;D^{g})}^{\mathrm{rank} _v(a;D^{g}) + D^{g}(a)} g(r) \;\mathrm{d}r  \right].
\]  

Multiplying \Cref{thm:gsl-key-property} by $D^g(a)$, we get
\[
		\Ev_{v \sim  V} \left[ \int_{\mathrm{rank} _v(a;D^{g})}^{\mathrm{rank} _v(a;D^{g}) + D^{g}(a)} g(r) \;\mathrm{d}r \right]  \leq  D^{g}(a) \int_{0}^{1} g(t) \;\mathrm{d}t.
\] 
Hence, $\widehat{D}(a) \leq D^{g}(a)$ for every candidate $a$. But $\widehat{D}$ and $D^g$ are both distributions, so we must have $\widehat{D} = D^{g}$. 
\end{proof}

Part of what makes $g$-stable lotteries an attractive tool is the vastness and flexibility of the design space they allow. Stable lotteries can already get a lot of mileage by introducing a discrete spectrum spanning from maximal lotteries at $k = 1$ to Random Dictatorship as $k \to \infty$. Not only can this spectrum be made continuous, but it is only a thin slice of possibilities in a broader function space. We will show that this space is large enough to reach voting rules with nearly optimal metric distortion.

\begin{mainrestatement}
There exists a continuous, non-decreasing function $g:[0,1]\to[0,1]$ with $g(0) = 0, g(1)=1$, such that in any election, every $g$-stable lottery has distortion at most $2.1248$.
\end{mainrestatement}

\section{Biased metrics}\label{sec:biased-metrics}

Our proofs make use of the \emph{biased metric} framework introduced by \citet{DBLP:conf/soda/CharikarR22}. At a high level, this framework gives a simple characterization of the worst-case metric spaces in a given election, independent of the voting rule. In doing so, it reduces the metric distortion objective down to its essence, into a more manageable (but still equivalent) form.  We will start with a broad overview of the framework, the intuition behind it, and the key reduction that it gives, but defer the proof details to prior work~\citep{DBLP:conf/soda/CharikarR22,DBLP:journals/jacm/CharikarRWW24,DBLP:conf/sigecom/CharikarRT025}.

Suppose we have an election with some underlying metric $d$ in which the candidate $a^\star$ has minimal social cost. Intuitively, if we modify the metric by \emph{decreasing} the distance from any voter $v$ to $a^\star$, and then \emph{increasing} the distance from $v$ to any candidate $b$ so that each difference $d(b, v) - d(a^\star, v)$ increases, the resulting transformation 
increases the distortion of any voting rule. While there could be a number of ways of making such a transformation, the biased metric framework shows that if we fix the distances from the optimal candidate to the other candidates, then there is a canonical transformation which results in the worst-case metrics. (See \Cref{fig:shorten-lengthen-schematic}.) The convenience is that rather than contending with all metrics which assign distances for every voter-candidate pair, the metrics we have to deal with are only defined by a function on candidates defining their distance to the optimum candidate. These metrics are defined below.\footnote{A note on notation: prior work used $x_a$ in place of $\rho(a)$, in part because these were originally thought of as variables in a linear program. Since we think of the candidates as occupying a spectrum of ranks, we hope the function notation is more natural.}

\begin{figure}[htbp]
\centering
\begin{tikzpicture}[
  candidate/.style={circle, fill=red!80!black, draw=red!45!black,
    inner sep=0pt, minimum size=6pt},
  voter/.style={circle, fill=RoyalBlue, draw=blue!55!black,
    inner sep=0pt, minimum size=6pt},
  metric edge/.style={draw=black!65, line width=0.85pt},
  voter edge/.style={draw=black!45, dashed, line width=0.85pt},
  shorten edge/.style={voter edge,
    postaction={decorate},
    decoration={markings,
      mark=at position 0.34 with
        {\arrow{Stealth[length=1.8mm,width=1.3mm]}},
      mark=at position 0.82 with
        {\arrowreversed{Stealth[length=1.8mm,width=1.3mm]}}}},
  lengthen edge/.style={voter edge,
    postaction={decorate},
    decoration={markings,
      mark=at position 0.28 with
        {\arrowreversed{Stealth[length=1.8mm,width=1.3mm]}},
      mark=at position 0.80 with
        {\arrow{Stealth[length=1.8mm,width=1.3mm]}}}},
  lengthen edge top/.style={voter edge,
    postaction={decorate},
    decoration={markings,
      mark=at position 0.24 with
        {\arrowreversed{Stealth[length=1.8mm,width=1.3mm]}},
      mark=at position 0.92 with
        {\arrow{Stealth[length=1.8mm,width=1.3mm]}}}},
  lengthen edge late/.style={voter edge,
    postaction={decorate},
    decoration={markings,
      mark=at position 0.33 with
        {\arrowreversed{Stealth[length=1.8mm,width=1.3mm]}},
      mark=at position 0.79 with
        {\arrow{Stealth[length=1.8mm,width=1.3mm]}}}},
  edge label/.style={fill=white, inner sep=1.5pt, font=\scriptsize},
  voter edge label/.style={edge label, text=black!55},
  point label/.style={font=\small}
]
  \coordinate (astar) at (0,0);
  \coordinate (c) at (199.1:5.0);
  \coordinate (cprime) at (319.1:2.5);
  \coordinate (b) at (79.1:1.50);
  \coordinate (v) at (-2.75,0.80);

  \draw[metric edge] (astar) --
    node[edge label, midway, sloped, above] {$\rho(c)$} (c);
  \draw[metric edge] (astar) --
    node[edge label, midway, sloped, above] {$\rho(c')$} (cprime);
  \draw[metric edge] (astar) --
    node[edge label, midway, right=2pt] {$\rho(b)$} (b);

  \draw[shorten edge] (v) --
    node[voter edge label, pos=0.58, sloped, above] {shorten} (astar);
  \draw[lengthen edge] (v) --
    node[voter edge label, pos=0.52, sloped, above] {lengthen} (c);
  \draw[lengthen edge late] (v) --
    node[voter edge label, pos=0.56, sloped, below] {lengthen} (cprime);
  \draw[lengthen edge top] (v) --
    node[voter edge label, pos=0.58, sloped, above] {lengthen} (b);

  \node[candidate] at (astar) {};
  \node[point label, red!80!black, anchor=west]
    at (0.13,0.11) {$a^\star$};
  \node[candidate] at (c) {};
  \node[point label, red!80!black, anchor=east]
    at (-4.85,-1.64) {$c$};
  \node[candidate] at (cprime) {};
  \node[point label, red!80!black, anchor=west]
    at (2.02,-1.64) {$c'$};
  \node[candidate] at (b) {};
  \node[point label, red!80!black, above=2pt] at (b) {$b$};
  \node[voter] at (v) {};
  \node[point label, RoyalBlue, anchor=south east]
    at (-2.85,0.90) {$v$};
\end{tikzpicture}
\caption{An illustration of the idea behind the biased metrics. Fix the distances from candidates to the optimal candidate ($a^\star$), defined by the function $\rho$. Then, minimize the distance from $v$ to $a^\star$, and maximize the distance from $v$ to other candidates.}
\label{fig:shorten-lengthen-schematic}
\end{figure}

\begin{definition}\label{def:biased}
Let $\rho:C \to \mathbb{R}_{\geq 0}$ be a function such that $\rho(a^\star) = 0$ for some $a^\star \in C$. 
Given an election instance, the \emph{biased metric} $d(\cdot, \cdot)$ corresponding to $\rho$ is defined as follows. For a voter $v$ and candidate $b$,
\begin{align*}
d(a^\star, v) &= \frac{1}{2} \max_{c \succeq_v c'} (\rho(c) - \rho(c'))\\
d(b, v) - d(a^\star, v) &= \min_{c \preceq_v b} \rho(c).
\end{align*}
\end{definition}

It may not be immediate to see why these metrics are sufficient, or even that they actually define metrics in the first place. We refer the reader to \citet{DBLP:conf/soda/CharikarR22} (or later expositions in \citet{DBLP:journals/jacm/CharikarRWW24}) for detailed proofs, but we offer a little intuition to demystify the terms in the definition. 

For $d(a^\star, v)$, one can think of $\rho$ as defining the preference that the optimal candidate $a^\star$ has for the other candidates. If a voter $v$ \emph{disagrees} with $a^\star$ over a pair of candidates, meaning that $c \succ_v c'$ but $\rho(c) > \rho(c')$, then it is not hard to see that $v$ cannot be collocated with $a^\star$. Indeed, this disagreement constrains just how close $v$ could be to $a^\star$, and with the triangle inequality one can find that this distance must be at least $\frac12(\rho(c) - \rho(c'))$. Considering the strongest constraint imposed by any pair of candidates, we get exactly the expression for $d(a^\star, v)$ in \Cref{def:biased}.

For $d(b, v)$, the triangle inequality immediately implies that this must be at most $ d(a^\star, v) + d(a^\star, b) = d(a^\star, v)  + \rho(b)$. But if $v$ prefers $b$ over some candidate $c$ with a smaller $\rho$, then we can instead apply the same inequality to $c$, and get an even stronger upper bound for $b$:
$$d(b, v) \leq d(c, v) \leq d(a^\star, v)  + \rho(c).$$
We again get the expression in \Cref{def:biased} by taking the choice of $c$ that results in the strongest constraint.

Next, we make formal the reduction that the biased metrics give us. For ease of exposition, 
 we introduce the shorthand 
$$\lowrho_v(b) = \min_{c\preceq_v b} \rho(c)$$
which corresponds to $d(b, v) - d(a^\star, v)$. Averaging over voters, we get 
$$\Ev_{b\sim D}[\SC(b) - \SC(a^\star)] = \Ev_{\substack{v\sim V\\b\sim D}}[\lowrho_v(b)].$$
The shorthand also simplifies the expression for $d(a^\star, v)$ since 
$$\max_{c \succeq_v c'} (\rho(c) - \rho(c')) = \max_b(\rho(b) - \lowrho_v(b)).$$
It follows that 
$$2\cdot \SC(a^\star) = \Ev_{v\sim V}[\max_b(\rho(b) - \lowrho_v(b))]$$
For a lottery $D$, write
$$L_v(D) = \Ev_{b\sim D}[\lowrho_v(b)], \qquad R_v = \max_b (\rho(b) - \lowrho_v(b)).$$
Thus $L_v(D)$ is the expected excess distance over $a^\star$, and $R_v$ is twice the distance to $a^\star$ in the biased metric. We write $L_v$ when $D$ is fixed.

Putting these together, the result of \citet{DBLP:conf/soda/CharikarR22} that the biased metrics are necessary and sufficient for bounding metric distortion can be summarized as the following theorem.

\begin{theorem}\label{thm:biased-iff}
For a distribution $D$ over candidates, $$\Ev_{b\sim D}[\SC(b)] \leq (1 + 2\lambda)\SC(a^\star)$$ for all consistent metrics if and only if for all $\rho:C\to\mathbb{R}_{\geq 0}$ such that $\rho(a^\star) = 0$, we have 
\begin{equation}\label{eq:precise-biased}
\Ev_{v\sim V}\left[L_v(D)\right] \leq \lambda \Ev_{v\sim V}\left[R_v\right]. 
\end{equation}
\end{theorem}

Finally, we conclude this section with a visual representation of the terms in \Cref{eq:precise-biased} (see \Cref{fig:rank-graph}) which we find useful for reasoning about them. 

Suppose that we fix a voter $v$ and a distribution $D$ over candidates. Similarly to \Cref{fig:rank-line}, partition the interval $[0, 1]$ into blocks 
for each candidate such that the length of a candidate's interval is their probability mass in $D$, and the intervals are ordered from $0$ to $1$ by $v$'s preference. That is, candidate $b$ occupies the interval $[\rank_v(b;D), \rank_v(b;D) + D(b))$. Now on top of the interval for each candidate $b$, add a box with height $\rho(b)$. The tops of these boxes trace out a function $\rho_v:[0, 1] \to \mathbb{R}_{\geq 0}$ which maps each rank $r$ to the $\rho$ value of the candidate occupying rank $r$. 

The running minimum of $\rho_v$ corresponds to $\lowrho_v$, and the shaded area in \Cref{fig:rank-graph} is precisely $L_v(D)$, $v$'s contribution to the left side of \Cref{eq:precise-biased}. Similarly, $v$'s contribution to the right side is $R_v$, which is the largest gap between the shaded region and the $\rho_v$. (Note that both the running minimum and the $\max_b$ include candidates whose block width is zero.)

\begin{figure}[ht]
\centering
\begin{tikzpicture}[x=12cm,y=12cm]
  \definecolor{lineA}{RGB}{203,125,177}
  \definecolor{lineC}{RGB}{86,180,233}
  \definecolor{lineD}{RGB}{0,158,115}
  \definecolor{lineE}{RGB}{230,159,0}
  \definecolor{lineF}{RGB}{0,114,178}
  \definecolor{lineB}{RGB}{213,94,0}
  \definecolor{lineG}{RGB}{204,121,167}
  \definecolor{lineH}{RGB}{112,160,205}
  \definecolor{lineI}{RGB}{95,185,145}
  \definecolor{lineJ}{RGB}{145,120,210}

  \definecolor{pastelorange}{RGB}{235,145,95}
  \definecolor{pastelbrown}{RGB}{242,231,216}
  \definecolor{shadeOrange}{RGB}{228,106,24}
  \colorlet{shadeFill}{pastelbrown!65!white}

  \fill[shadeFill] (0.00,0) rectangle (0.04,0.18);
  \fill[shadeFill] (0.04,0) rectangle (0.21,0.30);
  \fill[shadeFill] (0.21,0) rectangle (0.27,0.24);
  \fill[shadeFill] (0.27,0) rectangle (0.40,0.34);
  \fill[shadeFill] (0.40,0) rectangle (0.45,0.15);
  \fill[shadeFill] (0.45,0) rectangle (0.63,0.29);
  \fill[shadeFill] (0.63,0) rectangle (0.70,0.22);
  \fill[shadeFill] (0.70,0) rectangle (0.82,0.12);
  \fill[shadeFill] (0.82,0) rectangle (0.86,0.20);

  \path[
    preaction={fill=shadeFill,draw=none},
    pattern={Lines[angle=45,distance=4pt,line width=0.4pt]},
    pattern color=shadeOrange!85!black
  ] (0.00,0) rectangle (0.40,0.18);

  \path[
    preaction={fill=shadeFill,draw=none},
    pattern={Lines[angle=45,distance=4pt,line width=0.4pt]},
    pattern color=shadeOrange!85!black
  ] (0.40,0) rectangle (0.70,0.15);

  \path[
    preaction={fill=shadeFill,draw=none},
    pattern={Lines[angle=45,distance=4pt,line width=0.4pt]},
    pattern color=shadeOrange!85!black
  ] (0.70,0) rectangle (0.86,0.12);

  \draw[black,->,line width=0.65pt]
    (0,-0.011) -- (0,0.38);

  \draw[lineA,line width=4pt] (0.00,0) -- (0.04,0);
  \draw[lineC,line width=4pt] (0.04,0) -- (0.21,0);
  \draw[lineD,line width=4pt] (0.21,0) -- (0.27,0);
  \draw[lineE,line width=4pt] (0.27,0) -- (0.40,0);
  \draw[lineF,line width=4pt] (0.40,0) -- (0.45,0);
  \draw[lineB,line width=4pt] (0.45,0) -- (0.63,0);
  \draw[lineG,line width=4pt] (0.63,0) -- (0.70,0);
  \draw[lineH,line width=4pt] (0.70,0) -- (0.82,0);
  \draw[lineI,line width=4pt] (0.82,0) -- (0.86,0);
  \draw[lineJ,line width=4pt] (0.86,0) -- (1.00,0);

  \draw[black,line width=0.7pt]
    (0.00,0.18) -- (0.04,0.18) --
    (0.04,0.30) -- (0.21,0.30) --
    (0.21,0.24) -- (0.27,0.24) --
    (0.27,0.34) -- (0.40,0.34) --
    (0.40,0.15) -- (0.45,0.15) --
    (0.45,0.29) -- (0.63,0.29) --
    (0.63,0.22) -- (0.70,0.22) --
    (0.70,0.12) -- (0.82,0.12) --
    (0.82,0.20) -- (0.86,0.20) --
    (0.86,0) -- (1.00,0);

  \node[above=3pt] at (0.335,0.34) {$\rho_v(\cdot)$};

  \draw[pastelorange!80!black,line width=0.65pt]
    (0.00,0.18) -- (0.40,0.18) --
    (0.40,0.15) -- (0.70,0.15) --
    (0.70,0.12) -- (0.86,0.12) --
    (0.86,0) -- (1.00,0);

  \node[
    fill=shadeFill,
    rounded corners=1pt,
    inner sep=2pt
  ] at (0.21,0.085)
    {$\displaystyle \Ev_{b\sim D}\!\left[\lowrho_v(b)\right]$};

  \draw[black,<->,line width=0.6pt]
    (0.535,0.008) -- (0.535,0.282)
    node[
      midway,
      left=2pt,
      fill=shadeFill,
      rounded corners=1pt,
      inner sep=1pt
    ] {$\rho(b)$};

  \draw[black,<->,line width=0.6pt]
    (0.600,0.008) -- (0.600,0.142)
    node[
      midway,
      right=2pt,
      fill=shadeFill,
      rounded corners=1pt,
      inner sep=1pt
    ] {$\lowrho_v(b)$};

  \draw[black,line width=1.2pt] (1,-0.011) -- (1,0.011);
  \def\labelY{-0.035}
  \node[anchor=base] at (0.00,\labelY) {$0$};
  \node[anchor=base] at (0.54,\labelY) {$b$};
  \node[anchor=base] at (1.00,\labelY) {$1$};
\end{tikzpicture}
\caption{The function $\rho_v$ on voter $v$'s rank line.
Each candidate $b$ corresponds to a box of width $D(b)$ and height
$\rho(b)$. The hatched area under the running minimum
$\lowrho_v$ equals $\displaystyle\Ev_{b\sim D}[\lowrho_v(b)]$.}
\label{fig:rank-graph}
\end{figure}

Given the discussion above, and our way of treating each rank on the interval $[0, 1]$ as a copy of the candidate occupying that rank, with slight abuse of notation, we will treat the functions $\rho_v$ and $\lowrho_v$ as taking either candidates or ranks $r \in [0, 1]$ as input. For example, $\rho_v(b) = \rho_v(r)$ if candidate $b$ occupies rank $r$. As a reminder, this would make $\rho(b) = \rho_v(b)$ with $b\sim D$ and $\rho_v(r)$ with $r\sim \Unif(0, 1)$ equivalent.

\section{Warm-up: Distortion of stable \texorpdfstring{$k$}{k}-lotteries}\label{sec:stable-k}

We will start by giving an upper bound on the distortion of stable $k$-lotteries. The matching lower bound is proven in \Cref{app:lower-bound-stable-k} by \Cref{thm:lower-bound-stable-k}.

\begin{theorem}\label{thm:dist-sl}
In any election, stable $k$-lotteries have distortion at most $1 + 2\lambda_k$ where 
$$\lambda_k =  \frac{1 + (k-1)(k+1)^{-1/(k-1)}}{k}$$
and $\lambda_1 = 1$.
In particular, stable $2$-lotteries and stable $3$-lotteries have distortion $7/3 \approx 2.333.$
\end{theorem}

\begin{figure}[H]
\begingroup
\def\showstablekvalues{1}
\begin{tikzpicture}
    \begin{axis}[
        width=15cm,
        height=7.2cm,
        axis lines=left,
        axis line style={-stealth, line width=0.6pt},
        xmin=0.6, xmax=20.5,
        ymin=2.025, ymax=3.10,
        xlabel={$k$},
        ylabel={Distortion},
        label style={font=\small},
        tick label style={font=\small},
        xtick={1,...,20},
        ytick={2.2,2.4,2.6,2.8,3.0},
        yticklabel style={/pgf/number format/fixed,
                          /pgf/number format/precision=1,
                          /pgf/number format/fixed zerofill},
        tick align=outside,
        major tick length=3pt,
        xmajorgrids=false,
        ymajorgrids=true,
        major grid style={gray!22, line width=0.3pt},
        extra y ticks={2.11264},
        extra y tick labels={$2.11264$},
        extra y tick style={grid=none, major tick length=5pt},
        clip=false,
        /tikz/value label/.style={font=\scriptsize, text=violet!80!black,
                            inner sep=0.5pt},
    ]
        \addplot[black!65, dashed, line width=0.8pt]
            coordinates {(0.6,2.11264) (20.35,2.11264)};
        \node[anchor=south east, font=\small, text=black!75,
              inner sep=2pt] at (axis cs:20.25,2.125) {Lower bound};

        \addplot[violet, line width=1pt, no marks,
                 domain=1:20, samples=600]
            {1 + 2*(1 + (x-1)*(x+1)^(-1/max(x-1,0.001)))/x};

        \addplot[violet, only marks, mark=*, mark size=2.3pt,
                 mark options={solid, fill=violet, draw=white, line width=0.4pt}]
            coordinates {
            (1,3.0000000000)
            (2,2.3333333333)
            (3,2.3333333333)
            (4,2.3772053215)
            (5,2.4223089668)
            (6,2.4626848557)
            (7,2.4978973392)
            (8,2.5285499224)
            (9,2.5553674833)
            (10,2.5789934608)
            (11,2.5999584398)
            (12,2.6186908759)
            (13,2.6355357101)
            (14,2.6507717697)
            (15,2.6646259979)
            (16,2.6772845823)
            (17,2.6889015258)
            (18,2.6996052322)
            (19,2.7095035820)
            (20,2.7186878682)
            };

        \node[anchor=east, text=violet, font=\small]
            at (axis cs:20.2,2.94) {Stable $k$-lottery bound: $1+2\lambda_k$};

        \draw[violet!75!black, line width=0.5pt]
            (axis cs:2,2.310) -- (axis cs:2,2.275)
            -- (axis cs:3,2.275) -- (axis cs:3,2.310);
        \node[anchor=north, font=\small, text=violet!80!black,
              inner sep=2pt] at (axis cs:2.5,2.272) {$7/3\approx2.333$};

        \ifnum\showstablekvalues=1
        \node[value label,above=5pt] at (axis cs:1,3.0000000000) {$3$};
        \node[value label,below=5pt] at (axis cs:4,2.3772053215) {$2.377$};
        \node[value label,above=5pt] at (axis cs:5,2.4223089668) {$2.422$};
        \node[value label,below=5pt] at (axis cs:6,2.4626848557) {$2.463$};
        \node[value label,above=5pt] at (axis cs:7,2.4978973392) {$2.498$};
        \node[value label,below=5pt] at (axis cs:8,2.5285499224) {$2.529$};
        \node[value label,above=5pt] at (axis cs:9,2.5553674833) {$2.555$};
        \node[value label,below=5pt] at (axis cs:10,2.5789934608) {$2.579$};
        \node[value label,above=5pt] at (axis cs:11,2.5999584398) {$2.600$};
        \node[value label,below=5pt] at (axis cs:12,2.6186908759) {$2.619$};
        \node[value label,above=5pt] at (axis cs:13,2.6355357101) {$2.636$};
        \node[value label,below=5pt] at (axis cs:14,2.6507717697) {$2.651$};
        \node[value label,above=5pt] at (axis cs:15,2.6646259979) {$2.665$};
        \node[value label,below=5pt] at (axis cs:16,2.6772845823) {$2.677$};
        \node[value label,above=5pt] at (axis cs:17,2.6889015258) {$2.689$};
        \node[value label,below=5pt] at (axis cs:18,2.6996052322) {$2.700$};
        \node[value label,above=5pt] at (axis cs:19,2.7095035820) {$2.710$};
        \node[value label,below=5pt] at (axis cs:20,2.7186878682) {$2.719$};
        \fi
    \end{axis}
\end{tikzpicture}
\endgroup
\caption{Distortion of stable $k$-lotteries, compared to the lower bound of \citet{DBLP:conf/soda/CharikarR22}.}
\end{figure}

Before diving into the proof, we will give some high-level intuition for the strategy. The first idea is to take the distortion objective voter by voter. It would be a dream if we could show something like: 
\begin{equation}\label{eq:dream-no-potential}
L_v(D)  \leq \lambda_k R_v
\end{equation}
for each voter $v$, because then averaging over $v$, we would get distortion $1 + 2\lambda_k$ by \Cref{thm:biased-iff}. Of course, this dream is impossible to satisfy for every voter. The worst case situation is if $\rho(b)$ tends to decrease with $v$'s preferences (meaning the function $\rho_v$ depicted in \Cref{fig:rank-graph} is decreasing). In this case, then the right side of \Cref{eq:dream-no-potential} is close to zero (meaning that $v$ is close to the optimal candidate $a^\star$). 

To deal with this issue, the key idea is to offset the left side of \Cref{eq:dream-no-potential} with something that is negative in these worst-case situations, positive in the best-case situations, and averages out to zero over all of the voters thanks to the stationary property of a stable $k$-lottery (\Cref{prop:sl-stationary}).  

Consider two experiments: drawing one sample $b\sim D$, and drawing $k + 1$ samples $b_1, \dots, b_{k + 1} \sim D$ and taking $v$'s favorite (call this $\hat{b}_v \sim \hat{D}_v$). We know by \Cref{prop:sl-stationary} that when we average over the voters, these two distributions are the same. However, when we specialize this distribution for a particular voter $v$, we should expect $\hat{b}_v$ to be a lot better than $b$. If $\rho$ decreases with $v$'s preference, we expect that $\rho(\hat{b}_v)$ is smaller than $\rho(b)$. Thus, something like $\rho(\hat{b}_v) - \rho(b)$ could be a good choice for our offset. However, we can extract more value out of this charging scheme if we define some (increasing) potential function $\Phi$ that complements the existing terms in \Cref{eq:dream-no-potential} and instead use the offset
$$\Phi(\rho(\hat{b}_v)) - \Phi(\rho(b)).$$
Incorporating this idea into \Cref{eq:dream-no-potential}, we can show the lemma below.

\begin{lemma}\label{lem:key-sl-dist}
Fix an election, a biased metric parametrized by $\rho(\cdot)$, and an arbitrary distribution $D$ over candidates such that $\rho(b) = 0$ for some $b$ in the support of $D$.
Fix a voter $v$, and let $\widehat{D}_v$ be the distribution over candidates obtained by sampling $k + 1$ candidates from $D$ and outputting $v$'s favorite sample. There exists a potential function $\Phi$ depending only on $\rho(\cdot)$ and $D$ such that
\begin{equation}\label{eq:key-lem-sl}
L_v(D) + \Ev_{\hat{b}_v\sim \widehat{D}_v}[\Phi(\rho(\hat{b}_v))] - \Ev_{b\sim D}[\Phi(\rho(b))] \leq \lambda_k R_v
\end{equation}
\end{lemma}

Before diving into the proof, we briefly discuss some of the conditions introduced in the lemma. First, the condition that $\Phi$ does not depend on $v$ ensures that if $D$ is a stable $k$-lottery, we have
\[
		\Ev_{v\sim  V} \Ev_{\hat{b}_v\sim \widehat{D}_v} [\Phi(\rho(\hat{b}_v))] = \Ev_{b \sim  D} [\Phi(\rho(b))],
\] 
and so the offset vanishes on average. Second, the additional condition $\rho(b) = 0$ for some $b$ in the support of $D$ is an artifact of the choice of $\Phi$ that gets forced by the proof, which turns out to be undefined close to 0 if this condition does not hold. While a stable $k$-lottery may not always satisfy this condition (since $\rho$ is arbitrary), we can still establish \Cref{thm:dist-sl} by applying the lemma with a distribution that mixes between a stable $k$-lottery and a point mass at the optimal candidate. (Crucially, this is why the lemma is not specific to a stable $k$-lottery, and instead applies for any $D$ with nonzero mass on a zero-$\rho$ candidate.) The details are after the proof of the lemma.

\begin{proof}[Proof of \Cref{lem:key-sl-dist}]
    Fix a voter $v$. To prove the lemma, we need to choose a potential $\Phi$ and understand how it changes when we replace an ordinary draw $b \sim  D$ by $\hat{b}_v$, the voter's favorite among $k+1$ draws. The main challenge in comparing $b$ with $\hat{b}_v$ directly is that they are different in two ways at once: we draw more candidates, and we also use $v$'s preference to select among them. We address this challenge by introducing an intermediate candidate $\underline{b}$, the candidate with the smallest $\rho$-value among the same $k+1$ draws, and call its corresponding distribution $\underline{D}$. Doing so replaces the original comparison with two simpler ones: $\underline{b}$ is chosen from the same sample as $\hat{b}_v$, while its $\rho$-value is easy to relate to that of an ordinary draw $b \sim  D$.

    \begin{figure}[ht]
\centering
\begin{tikzpicture}[x=12cm,y=12cm]
  \definecolor{lineA}{RGB}{203,125,177}
  \definecolor{lineC}{RGB}{86,180,233}
  \definecolor{lineD}{RGB}{0,158,115}
  \definecolor{lineE}{RGB}{230,159,0}
  \definecolor{lineF}{RGB}{0,114,178}
  \definecolor{lineB}{RGB}{213,94,0}
  \definecolor{lineG}{RGB}{204,121,167}
  \definecolor{lineH}{RGB}{112,160,205}
  \definecolor{lineI}{RGB}{95,185,145}
  \definecolor{lineJ}{RGB}{145,120,210}

  \definecolor{sampleyellow}{RGB}{248,230,55}
  \definecolor{pastelorange}{RGB}{235,145,95}
  \definecolor{pastelbrown}{RGB}{242,231,216}
  \definecolor{shadeOrange}{RGB}{228,106,24}
  \colorlet{shadeFill}{pastelbrown!65!white}

  \fill[shadeFill] (0.00,0) rectangle (0.04,0.18);
  \fill[shadeFill] (0.04,0) rectangle (0.21,0.30);
  \fill[shadeFill] (0.21,0) rectangle (0.27,0.24);
  \fill[shadeFill] (0.27,0) rectangle (0.40,0.34);
  \fill[shadeFill] (0.40,0) rectangle (0.45,0.15);
  \fill[shadeFill] (0.45,0) rectangle (0.63,0.29);
  \fill[shadeFill] (0.63,0) rectangle (0.70,0.22);
  \fill[shadeFill] (0.70,0) rectangle (0.82,0.12);
  \fill[shadeFill] (0.82,0) rectangle (0.86,0.20);

  \path[
    preaction={fill=shadeFill,draw=none},
    pattern={Lines[angle=45,distance=4pt,line width=0.4pt]},
    pattern color=shadeOrange!85!black
  ] (0.00,0) rectangle (0.40,0.18);

  \path[
    preaction={fill=shadeFill,draw=none},
    pattern={Lines[angle=45,distance=4pt,line width=0.4pt]},
    pattern color=shadeOrange!85!black
  ] (0.40,0) rectangle (0.70,0.15);

  \path[
    preaction={fill=shadeFill,draw=none},
    pattern={Lines[angle=45,distance=4pt,line width=0.4pt]},
    pattern color=shadeOrange!85!black
  ] (0.70,0) rectangle (0.86,0.12);

  \draw[black,->,line width=0.65pt]
    (0,-0.011) -- (0,0.38);

  \draw[lineA,line width=4pt] (0.00,0) -- (0.04,0);
  \draw[lineC,line width=4pt] (0.04,0) -- (0.21,0);
  \draw[lineD,line width=4pt] (0.21,0) -- (0.27,0);
  \draw[lineE,line width=4pt] (0.27,0) -- (0.40,0);
  \draw[lineF,line width=4pt] (0.40,0) -- (0.45,0);
  \draw[lineB,line width=4pt] (0.45,0) -- (0.63,0);
  \draw[lineG,line width=4pt] (0.63,0) -- (0.70,0);
  \draw[lineH,line width=4pt] (0.70,0) -- (0.82,0);
  \draw[lineI,line width=4pt] (0.82,0) -- (0.86,0);
  \draw[lineJ,line width=4pt] (0.86,0) -- (1.00,0);

  \draw[black,line width=0.7pt]
    (0.00,0.18) -- (0.04,0.18) --
    (0.04,0.30) -- (0.21,0.30) --
    (0.21,0.24) -- (0.27,0.24) --
    (0.27,0.34) -- (0.40,0.34) --
    (0.40,0.15) -- (0.45,0.15) --
    (0.45,0.29) -- (0.63,0.29) --
    (0.63,0.22) -- (0.70,0.22) --
    (0.70,0.12) -- (0.82,0.12) --
    (0.82,0.20) -- (0.86,0.20) --
    (0.86,0) -- (1.00,0);

  \draw[pastelorange!80!black,line width=0.65pt]
    (0.00,0.18) -- (0.40,0.18) --
    (0.40,0.15) -- (0.70,0.15) --
    (0.70,0.12) -- (0.86,0.12) --
    (0.86,0) -- (1.00,0);

  \draw[black,<->,line width=0.6pt]
    (0.441,0.008) -- (0.441,0.142)
    node[
      midway,
      left=2pt,
      fill=shadeFill,
      rounded corners=1pt,
      inner sep=1pt,
      text=black
    ] {$\rho(\low{b})$};

  \draw[black,<->,line width=0.6pt]
    (0.829,0.008) -- (0.829,0.192)
    node[
      midway,
      left=2pt,
      yshift=-3pt,
      fill=shadeFill,
      rounded corners=1pt,
      inner sep=1pt,
      text=black
    ] {$\rho(\hat{b}_v)$};

  \fill[sampleyellow] (0.307,0) circle[radius=0.011];
  \fill[sampleyellow] (0.441,0) circle[radius=0.011];
  \fill[sampleyellow] (0.566,0) circle[radius=0.011];
  \fill[sampleyellow] (0.829,0) circle[radius=0.011];

  \draw[black,line width=1.2pt] (1,-0.011) -- (1,0.011);
  \def\labelY{-0.035}
  \node[anchor=base] at (0.00,\labelY) {$0$};
  \node[anchor=base] at (1.00,\labelY) {$1$};
\end{tikzpicture}
\caption{A visual depiction of $\rho(\low{b})$ and
$\rho(\hat{b}_v)$ following \Cref{fig:rank-graph}.
The yellow points represent $k + 1$ samples from $D$.
The sample furthest to the right corresponds to $\hat{b}_v$,
and the sample with the lowest $\rho$-value corresponds to $\low{b}$.}
\end{figure}

    \paragraph{Step 1. Using $b$ and $\underline{b}$ to choose $\Phi$.} Adding and subtracting the expected potential of $\underline{b}$, we can write the LHS of \Cref{eq:key-lem-sl} as 
    \begin{equation}\label{eq:key-lem-sl-expanded}
      L_v(D) + \left(\Ev_{\hat{b}_v\sim \widehat{D}_v}[\Phi(\rho(\hat{b}_v))] - \Ev_{\low{b}\sim \low{D}}[\Phi(\rho(\low{b}))]\right) - \left(\Ev_{b\sim D}[\Phi(\rho(b))] - \Ev_{\low{b}\sim \low{D}}[\Phi(\rho(\low{b}))]\right).
    \end{equation}
    
    We begin with the last term, the comparison between $b$ and $\underline{b}$. Since neither candidate is selected using $v$'s preferences, this term depends only on the distribution of their $\rho$-values. Let $F(t) = \Pr_{b\sim D}[\rho(b) \leq t]$. An ordinary draw $b$ has $\rho(b) > t$ with probability $1 - F(t)$. On the other hand, $\underline{b}$ has $\rho(\underline{b}) > t$ exactly when all $k+1$ samples have $\rho$-value $>t$. Therefore,
    \[
    		\Pr_{\low{b}\sim \low{D}}[\rho(\low{b}) > t] = (1 - F(t))^{k+1}.
    \] 
    For now, we leave $\Phi$ unspecified and write 
    \[
    		\Phi(z) = \int_{0}^{z} \Phi'(t) \;\mathrm{d}t.
    \] 
    Taking the expectations by layer-cake formula, we have
    \begin{align*}
    		\Ev_{b\sim D}[\Phi(\rho(b))] - \Ev_{\low{b}\sim \low{D}}[\Phi(\rho(\low{b}))]  &= \int_{0}^{\infty} \Phi'(t) \Pr_{b\sim D}[\rho(b) > t] \;\mathrm{d}t - \int_{0}^{\infty} \Phi'(t) \Pr_{\low{b}\sim \low{D}}[\rho(\low{b}) > t] \;\mathrm{d}t\\
            &= \int_{0}^{\infty} \Phi'(t) (1-F(t)) \;\mathrm{d}t - \int_{0}^{\infty} \Phi'(t) (1-F(t))^{k+1} \;\mathrm{d}t \\
    		&=\int_0^\infty \Phi'(t)(1 - F(t))(1-(1 - F(t))^k) \dif t.
    \end{align*}
    
    At this point, we want to choose $\Phi'(t)$ to cancel the factor $1 - (1-F(t))^{k}$. This simplifies the expression and begins to bring \Cref{eq:key-lem-sl-expanded} into the form of the RHS of \Cref{eq:key-lem-sl}. In particular, if we let
    \[
    		\Phi'(t) = \frac{1}{1-(1-F(t))^{k}},
    \] 
    then
    \begin{equation}
            \label{eq:sl-potential-cancellation}
    		\Ev_{b\sim D}[\Phi(\rho(b))] - \Ev_{\low{b}\sim \low{D}}[\Phi(\rho(\low{b}))] = \int_0^\infty (1 - F(t))\dif t = \Ev_{b\sim D}[\rho(b)],
    \end{equation}
    and \Cref{eq:key-lem-sl-expanded} simplifies into
    \begin{equation}\label{eq:key-lem-sl-expanded-simplified}
      \left(\Ev_{\hat{b}_v\sim \widehat{D}_v}[\Phi(\rho(\hat{b}_v))] - \Ev_{\low{b}\sim \low{D}}[\Phi(\rho(\low{b}))]\right) - \Ev_{b\sim D}[\rho(b) - \lowrho_v(b)].
    \end{equation}

    Before continuing, we make two important observations about $\Phi'(t)$, and by extension the definition of $\Phi$ as 
    \[
    		\Phi(z) = \int_{0}^{z} \frac{\mathrm{d} t}{1-(1-F(t))^{k}}.
    \]
    Firstly, to avoid the denominator being 0, we need that $F(t) > 0$ for all $t$, meaning that $\Pr_{b\sim D}[\rho(b) = 0]$ is positive. (We remove this assumption when $D$ is a stable $k$-lottery after the proof of the lemma.) Second, we observe that $\Phi$ truly only depends on $\rho(\cdot)$ and $D$, and does not depend on $v$'s preferences, since the same is true of $F$.

    \paragraph{Step 2. Comparing $\hat{b}_v$ with $\underline{b}$.} Continuing, our goal now turns to bounding
    \[
    		\Ev_{\hat{b}_v\sim \widehat{D}_v}[\Phi(\rho(\hat{b}_v))] - \Ev_{\low{b}\sim \low{D}}[\Phi(\rho(\low{b}))].
    \] 
    Recall that to sample $b \sim  D$, it is equivalent to sample $r \sim  [0,1]$ uniformly and choose the candidate whose rank interval $[\rank_v(b;D), \rank_v(b;D) + D(b))$ contains $r$. As in the previous section, let $\rho_v(r)$ denote the $\rho$-value of the candidate occupying rank $r$. Then sampling $\hat{b}_v\sim\widehat{D}_v$ corresponds to drawing $r_1,\dots,r_{k+1}\sim\Unif(0,1)$ and outputting the candidate at rank $\max_i r_i$, while sampling $\low{b}\sim\low{D}$ corresponds to the value $\min_i\rho_v(r_i)$.

    With this in mind, we have:
    \begin{equation}\label{eq:final-potentials}
    \Ev_{\hat{b}_v\sim \widehat{D}_v}[\Phi(\rho(\hat{b}_v))] - \Ev_{\low{b}\sim \low{D}}[\Phi(\rho(\low{b}))] =  \Ev_{r_1, \dots, r_{k+1} \sim \Unif(0,1)} \left[\Phi\bigl(\rho_v(\max_i r_i)\bigr) - \Phi\bigl(\min_i \rho_v(r_i)\bigr)\right].
    \end{equation}
    
    We can imagine drawing $r_1, \hdots, r_{k+1}$ in two stages. First, draw their largest rank, $s = \max_i r_i$. Then draw the remaining $k$ ranks independently and uniformly from $[0,s]$. The candidate at rank $s$ is $\hat{b}_v$, so fixing $s$ also fixes $\rho_v(s)$. The other $k$ ranks determine whether the smallest $\rho$-value among the samples falls below $\rho_v(s)$, and by how much. The largest rank is at most $s$ exactly when all $k+1$ ranks lie in $[0,s]$, which happens with probability $s^{k+1}$. Therefore, the largest rank has probability density $(k+1) s^k$ at $s$. It follows that the expression above is equal to
    \begin{align*}
    & \int_0^1 (k+1)s^k \Ev_{r_1, \dots, r_k\sim\Unif(0,s)}\left[\Phi(\rho_v(s)) - \Phi\bigl(\min_i \rho_v(r_i)\bigr)\right] \dif s
    \end{align*}
    (For ease of exposition, we can treat $s$ as $r_{k + 1}$, so that $\min_i \rho_v(r_i)$ includes $\rho_v(s)$.) Rearranging, we get
    \begin{align*}
    & \int_0^1 (k+1)s^k \Ev_{r_1, \dots, r_k\sim\Unif(0,s)}\left[\int_{\min_i \rho_v(r_i)}^{\rho_v(s)}\Phi'(t) \dif t \right] \dif s\\
    = &\int_0^1 (k+1)s^k \int_{0}^{\rho_v(s)}\Phi'(t)\Pr_{r_1, \dots, r_k\sim\Unif(0,s)} [\min_i \rho_v(r_i) \leq t] \dif t  \dif s
    \end{align*}
    
    To describe the other $k$ draws, let 
    \[
    		F_s(t) = \Pr_{r \sim \Unif(0,s)} [\rho_v(r)\leq t].
    \]
    If $t < \underline{\rho}_v(s)$, then no rank in $[0,s]$ has $\rho$-value at most $t$. Thus $F_s(t) = 0$, and the part of the integral below $\underline{\rho}_v(s)$ contributes nothing. For $t < \rho_v(s)$, we have $\rho_v(s) > t$, and hence $\min \{\rho_v(s), \rho_v(r_1), \hdots, \rho_v(r_k)\} \le t$ if and only if $\min_{1\le i \le k} \rho_v(r_i) \le t$. Therefore,
    \[
    		\Pr_{r_1, \dots, r_k\sim\Unif(0,s)} [\min_i \rho_v(r_i) \leq t] = 1 - (1 - F_s(t))^k.
    \]
    Putting everything together, we get
    \begin{align}
    \Ev_{\hat{b}_v\sim \widehat{D}_v}[\Phi(\rho(\hat{b}_v))] - \Ev_{\low{b}\sim \low{D}}[\Phi(\rho(\low{b}))] &=  \int_0^1 (k+1)s^k \int_{\lowrho_v(s)}^{\rho_v(s)}\Phi'(t) (1 - (1 - F_s(t))^k)\dif t  \dif s \notag \\
    &= \int_0^1 (k+1)s^k \int_{\lowrho_v(s)}^{\rho_v(s)}\frac{1 - (1 - F_s(t))^k}{1 - (1 - F(t))^k}\dif t  \dif s. \label{eq:final-potentials-ratio}
    \end{align}
    
    We now bound the integrand in \Cref{eq:final-potentials-ratio} and claim that, for $s>0$,
    \begin{equation}
            \label{eq:potential-1-over-s-bound}
    		\frac{1-(1-F_s(t))^{k}}{1-(1-F(t))^{k}} \le \frac{1}{s}.
    \end{equation}
    Indeed, since
    \[
    		F(t) = \Pr_{r\sim\Unif(0,1)}[\rho_v(r) \leq t] \geq s \Pr_{r\sim\Unif(0,s)}[\rho_v(r) \leq t] = s F_s(t),
    \] 
    and $u\mapsto 1 - (1-u)^{k}$ is increasing, concave, and vanishes at $u=0$, we have
    \[
    		1-(1-F(t))^{k} \ge 1 - (1-s F_s(t))^{k} \ge s (1-(1-F_s(t))^{k}).
    \] 
    The denominator is positive by Step 1, so the claim follows. 
    
    \paragraph{Step 3. Putting the bounds together.} Applying the $1 /s$ bound to \Cref{eq:final-potentials-ratio}, we obtain
    \begin{align*}
    \Ev_{\hat{b}_v\sim \widehat{D}_v}[\Phi(\rho(\hat{b}_v))] - \Ev_{\low{b}\sim \low{D}}[\Phi(\rho(\low{b}))] &\leq \int_0^1 (k+1)s^k \int_{\lowrho_v(s)}^{\rho_v(s)}\frac1s\dif t  \dif s\\
    &= \int_0^1 (k+1)s^{k-1} (\rho_v(s) - \lowrho_v(s)) \dif s.
    \end{align*}
    To combine this with $\Ev_{b\sim D}[\rho(b) - \lowrho_v(b)]$ in \Cref{eq:key-lem-sl-expanded-simplified}, observe that 
    \[
    		\Ev_{b\sim D}[\rho(b) - \lowrho_v(b)] = \int_0^1 (\rho_v(s) - \lowrho_v(s)) \dif s.
    \] 
    Putting all of these expressions together, we have
    \begin{align}
        \label{eq:sl-final-integral-bound}
        L_v(D) + \Ev_{\hat{b}_v\sim \widehat{D}_v}[\Phi(\rho(\hat{b}_v))] - \Ev_{b\sim D}[\Phi(\rho(b))] &\leq \int_0^1 ((k+1)s^{k-1} - 1) (\rho_v(s) - \lowrho_v(s)) \dif s  \nonumber \\
        &\leq \left(\int_0^1 ((k+1)s^{k-1} - 1)_+ \dif s\right)\cdot R_v.
    \end{align}

    Finally, the coefficient is positive when $s > (k+1)^{-1 /(k-1)}$ (naturally treated as $0$ if $k = 1$), so
    \begin{align*}
    		\int_{0}^{1} ((k+1) s^{k-1}-1)_+\;\mathrm{d}s &= \int_{(k+1)^{-1 /(k-1)}}^{1} ((k+1)s^{k-1} - 1) \;\mathrm{d}s \\
    		&= \frac{1 + (k-1)(k+1)^{-1 /(k-1)}}{k} = \lambda_k,
    \end{align*}
    as claimed. 
    \end{proof}

    Finally, we show how to establish \Cref{thm:dist-sl} using \Cref{lem:key-sl-dist}.

    \begin{proof}[Proof of \Cref{thm:dist-sl}.]
    Recall at the end of Step 1 that we assumed $D$ assigns positive probability to some candidate with $\rho = 0$. 
    Steps 1-3 establish our Lemma under this assumption, for averaging \Cref{eq:key-lem-sl} over $v$ and applying \Cref{prop:sl-stationary} gives \Cref{eq:precise-biased}. We now remove this assumption and extend the result to any stable $k$-lottery $D$, in particular, one where $\rho(b) > 0$ for every $b\in \mathrm{supp}(D)$. 

    Choose $a^\star$ with $\rho(a^\star) = 0$, and let 
    \[
            \beta = \min \{ \rho(b): b\in \mathrm{supp}(D)\} > 0, \qquad D_\eps = (1-\eps)D + \eps \delta_{a^\star}.
    \]
    Since $D_\eps (a^\star) > 0$, the Lemma applies to $D_\eps$. Let $\widehat{D}_\eps$ be the distribution obtained by drawing $v\sim V$, drawing $k+1$ candidates from $D_\eps$, and returning $v$'s favorite, and let $\Phi_\eps$ be the function obtained from Step 1 with $D_\eps$ in place of $D$.
    Adding a constant to $\Phi_\eps$ does not change \Cref{eq:key-lem-sl}, so we normalize it and set $\Phi_\eps(\beta) = 0$. Since $\underline{\rho}_v(a^\star) = 0$, we have $L_v(D_\eps) = (1-\eps)L_v(D)$. Averaging \Cref{eq:key-lem-sl} gives
    \begin{equation}
            \label{eq:key-lem-sl-approx}
            (1-\eps) \Ev_{v\sim V}\left[L_v(D)\right]
            + \Ev_{b\sim \widehat{D}_\eps}  [\Phi_\eps (\rho(b))] 
            - \Ev_{b\sim D_\eps} [\Phi_\eps (\rho(b))] 
            \leq  \lambda_k \Ev_{v\sim V}  \left[ R_v\right].
    \end{equation}
    By definition, $\mathrm{supp}(D_\eps) = \mathrm{supp}(D) \cup \{a^\star\}$. Moreover, $\widehat{D}_\eps$ returns one of the candidates drawn from $D_\eps$ and hence has the same support. Therefore, we can write the two terms involving $\Phi$ on the LHS as
    \[
            \Ev_{b\sim \widehat{D}_\eps}  [\Phi_\eps (\rho(b))] 
            - \Ev_{b\sim D_\eps} [\Phi_\eps (\rho(b))]  
            = (\widehat{D}_\eps(a^\star)-\eps)\Phi_\eps(0) + \sum_{b\in \mathrm{supp}(D)} (\widehat{D}_\eps (b) - D_\eps(b)) \Phi_\eps(\rho(b))
    \]

    As $\eps \to 0$, $D_\eps$ approaches $D$, and the same sampling procedure gives $\widehat{D}_\eps \to \widehat{D} = D$. Our normalization keeps $\Phi_\eps(\rho(b))$ bounded for $b\in \mathrm{supp}(D)$. Therefore, the second term involving $\sum_{b\in \mathrm{supp}(D)}$ tends to zero.

    It remains to consider the first term that involves $a^\star$. Let 
    \[
            \alpha = \Pr_{\substack{v\sim V \\ b_1, \hdots, b_k \sim D}} [a^\star \succ_v b_1, \hdots, b_k]
    \]
    and note that $\alpha \leq 1 / (k+1)$ by \Cref{eq:sl-key-property}. Let $F_\eps$ be the CDF of $\rho(b)$ for $b\sim D_\eps$. For $0 \leq t < \beta$, we have $F_\eps(t) = \eps$, so by definition
    \[
            \Phi_\eps(0) = - \frac{\beta}{1-(1-\eps)^k}.
    \]

    How large is $\widehat{D}_\eps (a^\star)$? If $a^\star$ appears exactly once among the $k+1$ samples, it is returned with probability $\alpha$. The probability that it appears at least twice is $O(\eps^2)$. Therefore,
    \[
            \widehat{D}_\eps (a^\star) = (k+1) \eps (1-\eps)^k \alpha + O(\eps^2).
    \]
    It follows that
    \[
            (\widehat{D}_\eps(a^\star) - \eps) \Phi_\eps(0) \to \frac{\beta}{k} (1-(k+1)\alpha) \geq 0.
    \]
    Letting $\eps \to 0$ and applying it to \Cref{eq:key-lem-sl-approx} implies \Cref{eq:precise-biased}, so our proof is complete!
    \end{proof}

\section{Optimal distortion of $g$-stable lotteries}
\label{sec:optimal-g}

So far, we showed that stable $k$-lotteries, which correspond to activations $g(r)=r^k$, achieve distortion $7/3\approx 2.333$ when $k\in\{2,3\}$. It is natural to ask how much further we can push the distortion bound when we optimize over all activation functions $g$.

Writing $\operatorname{distortion}(g)$ for the worst-case distortion of a $g$-stable lottery over all elections, in this section we prove the following result. 

\begin{theorem}\label{thm:optimal-g}
	The optimal distortion among $g$-stable lotteries satisfies
	\[
	2.1232\;\leq\;\inf_g \operatorname{distortion}(g)\;\leq\;2.1248,
	\]
	where the infimum ranges over all continuous, nondecreasing functions $g:[0,1]\to[0,1]$ that are not constant 
\end{theorem}

We note that the bounds in \Cref{thm:optimal-g} can be sharpened slightly. Using Lean we have been able to verify lower and upper bounds within $10^{-5}$,  $2.12465$ and $2.12466$ respectively, and we also have an AI-generated proof (without human authorship) which claims to exactly give $\inf_g \operatorname{distortion}(g) = 1 + 2\lambda^\star \approx 2.1246569.$ We include both in our \href{https://github.com/YQL-Skorpion/stable-lotteries}{GitHub repository}. Our goal in this section is to give a human-readable proof which retains the key insights and as much of the guarantee as possible (particularly in the upper bound, which is our primary focus).

The key idea is that we ask what difficult elections force upon \emph{any} activation function, and from these observations, we reverse-engineer a good $g$. In \Cref{sec:three-elections}, we construct three simple families of elections. The first two elections impose linear constraints on the score function of every activation, and we show that discretizing these constraints yields a linear system that is infeasible below $2.1232$, which gives our lower bound (\Cref{sec:lower-g}). The third election combines the first two and yields a differential inequality that guides the shape of an optimal score function at large ranks. For the upper bound (\Cref{sec:upper-g}), we prove the Verifier Lemma, which reduces the worst-case distortion bound over all elections to a purely analytic system of inequalities. 
We then exhibit an explicit solution to this system and verify it with the two scripts described in \Cref{app:optimal-g}. Both scripts and their certificate data are also available in \href{https://github.com/YQL-Skorpion/stable-lotteries/tree/main/section_6_checker}{the same repository}.

Throughout, we keep the notation of \Cref{sec:biased-metrics}. Fix an election, a lottery $D$, an optimal candidate $a^\star$, and distances $\rho:C\to\mathbb{R}_{\geq0}$ with $\rho(a^\star)=0$, and write
\[
L_v=L_v(D)=\mathbb{E}_{b\sim D}[\lowrho_v(b)],
\qquad
 R_v=\max_b\left(\rho(b)-\lowrho_v(b)\right),
\]
where both the running minimum $\lowrho_v(b)=\min_{c\preceq_v b}\rho(c)$ and the maximum over $b$ include candidates outside the support of $D$.  Recall, by \Cref{eq:precise-biased}, that a lottery $D$ has distortion at most $1+2\lambda$ if and only if
\begin{equation}\label{eq:height-goal}
	\mathbb{E}_v[L_v]\leq\lambda\mathbb{E}_v[R_v]
\end{equation}
for every choice of $a^\star$ and $\rho$.

\subsection{Three elections}
\label{sec:three-elections}

Suppose that $g$ guarantees distortion at most $1 + 2\lambda$. Because a $g$-stable lottery is defined implicitly through a minimax problem, starting from an election and solving for its $g$-stable lottery would be unwieldy. Instead, we work in reverse: we fix $D$ and distances $\rho$ first, and then design the voters' preferences so that $D$ is $g$-stable while making $\mathbb{E}_v[L_v]$ as large as possible relative to $\mathbb{E}_v[R_v]$.

How do we ensure that such a lottery $D$ is $g$-stable? Recall from \Cref{thm:gsl-key-property} and \Cref{prop:gsl-stationary} that a lottery $D$ is $g$-stable if and only if every candidate $a$ satisfies $\mathbb{E}_v[g(\mathrm{rank}_v(a;D))]\leq \tau=\int_0^1 g(r)\,dr$, with equality for every candidate in the support of $D$.
Because we keep comparing $g(r)$ against $\tau$, it is convenient to center $g$ around $\tau$. For a scaling factor $c>0$ that we will choose shortly, define
\begin{equation}\label{eq:score}
    \phi(r)=c\left(\tau-g(r)\right),
    \qquad
    J(r)=\int_0^r\phi(u)\,du.
\end{equation}
Some immediate properties of $\phi$ and $J$: $\phi$ is continuous, nonincreasing, and integrates to zero on $[0,1]$; it is positive at low ranks (where $g(r)<\tau$) and negative at high ranks (where $g(r)>\tau$); and consequently $J(0)=J(1)=0$, with $J$ strictly positive on $(0,1)$.

We refer to $-\phi(r)=c(g(r)-\tau)$ as the \emph{score} at rank $r$. Under the rank-line representation from \Cref{sec:prelims}, a candidate's score on a voter's rank line is the average of $-\phi$ over its interval. A candidate outside the support of $D$ that sits at rank $r$ has score $-\phi(r)$, while a block of candidates occupying the rank interval $[u,w]$ has total score
\[
-\int_u^w \phi(r)\,dr = J(u)-J(w).
\]
Notice in particular that a block $[0,p]$ at the very bottom of a voter's ranking has total score $-J(p)<0$, whereas a block $[1-p,1]$ at the very top has total score $J(1-p)>0$. Since $\int_0^1\phi=0$, the scores on every voter's rank line sum to zero.

Each of the three elections below is built around the voters that make \Cref{eq:height-goal} hardest to satisfy, namely, those with $L_v$ large and $R_v$ small.

Throughout these constructions, we only care about candidates through their $\rho$-values, so we group candidates in $\mathrm{supp}(D)$ that share the same $\rho$-value into a class. Since a voter may rank some candidates of a class low and others high, we ``symmetrize'' each class to treat its candidates alike: we split its probability mass under $D$ equally among many candidates and replicate each voter's ranking across all permutations of the class. Every candidate in a class then receives an equal share of the class's total score, so $D$ is $g$-stable if and only if each class in $\mathrm{supp}(D)$ has expected total score zero, and each candidate outside $\mathrm{supp}(D)$ has expected score at most $0$.

\paragraph{Election 1. Isolating $a^\star\notin\mathrm{supp}(D)$: a bound on $\phi$.}
Our first construction tests how $g$ responds when $a^\star$ is outside the support of $D$ and every candidate in $\mathrm{supp}(D)$ has $\rho(b)=1$. Every candidate in $\mathrm{supp}(D)$ therefore has expected score zero. Thus, $D$ is $g$-stable as long as the distribution of $a^\star$'s rank across voters gives $a^\star$ an expected score of at most zero.

Consider a voter $v$ who places $a^\star$ at rank $r\in[0,1]$. Every candidate ranked below $a^\star$ has $\lowrho_v(b)=1$, while every candidate ranked above $a^\star$ has $\lowrho_v(b)=\rho(a^\star)=0$.
\begin{itemize}
    \item If $r<1$, the candidates below $a^\star$ have total mass $r$, so $L_v=r$; meanwhile, the candidates in $\mathrm{supp}(D)$ ranked above $a^\star$ have $\rho(b)-\lowrho_v(b)=1$, so $R_v=1$.
    \item If $r=1$, meaning that $v$ ranks $a^\star$ above the entire support of $D$, then $L_v=1$, but $\lowrho_v(b)=\rho(b)=1$ for every $b\in\mathrm{supp}(D)$ and $\lowrho_v(a^\star)=\rho(a^\star)=0$, so $R_v=0$! (Geometrically, a voter whose favorite candidate is $a^\star$ can sit at distance zero from $a^\star$.)
\end{itemize}
Therefore, a voter who places $a^\star$ at rank $r$ contributes $f_\lambda(r)$ to $L_v-\lambda R_v$, where
\[
f_\lambda(r)=
\begin{cases}
    r-\lambda,&r<1,\\
    1,&r=1.
\end{cases}
\]

In order for \Cref{eq:height-goal} to hold, every distribution of ranks for $a^\star$ with expected score $\mathbb{E}_v[c(g(r)-\tau)]\leq 0$ must satisfy $\mathbb{E}_v[f_\lambda(r)]\leq 0$. We claim that we can choose $c>0$ in \Cref{eq:score} so that $f_\lambda(r)\leq c(g(r)-\tau)$ for all $r\in[0,1]$. Otherwise, we could mix voters at two ranks so that $a^\star$ has mean activation exactly $\tau$, but $\mathbb{E}_v[f_\lambda(r)]>0$, contradicting \Cref{eq:height-goal}. Fixing this choice of $c$ in \Cref{eq:score}, the inequality $f_\lambda(r)\leq -\phi(r)$ gives the following bounds for every $r\in[0,1]$:
	\begin{equation}\label{eq:ceiling}
		\phi(r)\leq\lambda-r,
		\qquad
		\phi(1)\leq-1.
	\end{equation}

The jump from $\lambda-1$ to $-1$ at $r=1$ reflects the ``danger'' of voters who rank $a^\star$ at the very top. Because they contribute $L_v=1$ while paying nothing toward $R_v$, the activation at rank $1$ must be large enough to cap the fraction of such voters. In retrospect, this explains why stable $k$-lotteries perform so well: for $g(r) = r^k$, the activation at rank $1$ is $1$ while the average activation is $1 / (k+1)$, so stability limits the fraction of such voters to at most $1/(k+1)$. As we will see in \Cref{sec:upper-g}, these same voters also form the critical case in our upper-bound analysis.

\begin{figure}[htbp]
\centering
\definecolor{forbid}{RGB}{248,215,210}
\definecolor{jumpcol}{RGB}{176,32,32}
\definecolor{ktwo}{RGB}{40,90,190}
\definecolor{kthree}{RGB}{20,140,100}
\definecolor{kone}{RGB}{120,120,120}
\begin{tikzpicture}[x=8cm,y=3.4cm,font=\small,lab/.style={font=\small},>={Latex[length=1.8mm,width=1.3mm]}]
  \fill[forbid] (0,2/3) -- (1,-1/3) -- (1,1.1) -- (0,1.1) -- cycle;
  \node[lab,text=jumpcol!80!black,align=center] at (0.62,0.78) {$\phi(r)>\lambda-r$:\\[2pt]excluded by Election~1};
  \draw[->,line width=0.65pt] (0,0) -- (1.07,0) node[below,lab] {$r$};
  \draw[->,line width=0.65pt] (0,-1.15) -- (0,1.15) node[left,lab] {$\phi(r)$};
  \draw (0.01,2/3) -- (-0.01,2/3) node[left,lab] {$\lambda=\frac23$};
  \draw (0.01,-1) -- (-0.01,-1) node[left,lab] {$-1$};
  \draw (1,0.03) -- (1,-0.03); \node[lab,anchor=north west] at (1,0) {$1$};
  \draw[black,line width=.9pt] (0,2/3) -- (1,-1/3);
  \node[lab,anchor=south west] at (0.1,0.63) {$\lambda-r$};
  \draw[jumpcol,densely dotted,line width=1pt] (1,-1/3) -- (1,-1);
  \draw[jumpcol,fill=white,line width=0.8pt] (1,-1/3) circle[radius=1.6pt];
  \fill[jumpcol] (1,-1) circle[radius=2pt];
  \node[lab,text=jumpcol,anchor=west,align=left] at (1.02,-0.67) {jump: $\phi(1)\le-1$\\ (voters ranking $a^\star$\\ above $\mathrm{supp}(D)$)};
  \draw[kone,dash pattern=on 3pt off 2.5pt,line width=.9pt] (0,1) -- (1,-1);
  \draw[ktwo,line width=1.2pt,domain=0:1,samples=60,smooth] plot (\x,{0.5-1.5*\x*\x});
  \draw[kthree,line width=1.2pt,domain=0:1,samples=60,smooth] plot (\x,{1/3-4/3*\x*\x*\x});
  \begin{scope}[shift={(0.05,-0.62)}]
    \draw[kone,dash pattern=on 3pt off 2.5pt,line width=.9pt] (0,0) -- (0.07,0); \node[lab,anchor=west] at (0.08,0) {$g(r)=r$ (maximal lotteries)};
    \draw[ktwo,line width=1.2pt] (0,-0.17) -- (0.07,-0.17); \node[lab,anchor=west] at (0.08,-0.17) {$g(r)=r^2$};
    \draw[kthree,line width=1.2pt] (0,-0.34) -- (0.07,-0.34); \node[lab,anchor=west] at (0.08,-0.34) {$g(r)=r^3$};
  \end{scope}
  \fill[ktwo] (1/3,1/3) circle[radius=1.8pt];
  \fill[kthree] (1/2,1/6) circle[radius=1.8pt];
  \draw[ktwo,densely dotted,line width=.45pt] (1/3,0) -- (1/3,1/3); \node[lab,text=ktwo,below] at (1/3,0) {$\frac13$};
  \draw[kthree,densely dotted,line width=.45pt] (1/2,0) -- (1/2,1/6); \node[lab,text=kthree,below] at (1/2,0) {$\frac12$};
\end{tikzpicture}
\caption{The constraints~\eqref{eq:ceiling} from Election~1 for $\lambda=\frac23$. Each activation $g$ is drawn through
its score function $\phi=c(\tau-g)$, scaled by $c=1/(g(1)-\tau)$ so that the jump condition $\phi(1)\le-1$ holds with
equality; a larger $c$ would only push $\phi$ further up at small ranks. The score must stay below the line
$\lambda-r$ on $[0,1)$ and drop to at most $-1$ at $r=1$, far below the line's endpoint $\lambda-1$. Stable
$2$- and $3$-lotteries ($g(r)=r^2$ and $g(r)=r^3$) touch the line at $r=\frac13$ and $r=\frac12$, so Election~1 alone
forces $\lambda\ge\frac23$ for them, matching the bound $\frac73$. The linear activation of maximal lotteries cannot
satisfy both conditions for $\lambda<1$: with $\phi(1)=-1$ it starts at $\phi(0)=1$.}
\label{fig:ceiling}
\end{figure}

\paragraph{Election 2. Two classes of voters: a bound on $J$.}
In the first construction, where $a^\star\notin\mathrm{supp}(D)$, a voter could achieve $R_v=0$ only by ranking $a^\star$ at the very top. If instead $D$ places probability mass $1-p$ on candidates with $\rho=0$ (including $a^\star$) and mass $p$ on candidates with $\rho=1$, a voter can achieve $R_v=0$ simply by ranking all the $\rho=1$ candidates below all the $\rho=0$ candidates.

To turn this observation into a constraint on $J$, we partition the support of $D$ into two classes: $\mathcal{H}$ (representing \underline{h}igh values of $\rho$), whose candidates have $\rho=1$ and total probability $p$ under $D$, and $\mathcal{Z}$ (representing \underline{z}ero values of $\rho$), whose candidates have $\rho=0$ and total probability $1-p$. First consider voters who prefer every candidate in $\mathcal{Z}$ to every candidate in $\mathcal{H}$. These voters have $L_v=p$ and $R_v=0$, so they make \Cref{eq:height-goal} harder to satisfy. However, they give $\mathcal{H}$ a negative total score of $-J(p)$, which must be balanced for $D$ to be stable.

To balance this with a positive score for $\mathcal{H}$, we introduce a second type of voter who places some candidates from $\mathcal{H}$ above $\mathcal{Z}$. Choose $0\leq q<p$. We let these voters put candidates from $\mathcal{H}$ with total probability $q$ at the bottom of their ranking, then all of $\mathcal{Z}$, and then the remaining candidates from $\mathcal{H}$. The three groups occupy the rank intervals $[0,q]$, $[q,1-p+q]$, and $[1-p+q,1]$, respectively. We summarize this construction in the table below, with each subscript giving the group's total probability under $D$:
\[
\begin{array}{c@{\qquad}c@{\quad}cc}
    \toprule
    \text{Preference order (worst to best)}&\text{Total score of }\mathcal{H}&L_v&R_v\\
    \midrule
    \mathcal{H}_q,\ \mathcal{Z}_{1-p},\ \mathcal{H}_{p-q}
    &J(1-p+q)-J(q)&q&1\\
    \mathcal{H}_p,\ \mathcal{Z}_{1-p}&-J(p)&p&0\\
    \bottomrule
\end{array}
\]
For the voters who split $\mathcal{H}$, the running minimum $\lowrho_v$ drops from $1$ to $0$ as soon as $\mathcal{Z}$ appears at rank $q$, so only the initial block $\mathcal{H}_q$ contributes to $L_v=q$, while the top block $\mathcal{H}_{p-q}$ gives $R_v=1$.

Let $\delta_p(q)=J(1-p+q)-J(q)$ denote the total score of $\mathcal{H}$ for voters who split it around $\mathcal{Z}$. Whenever $\delta_p(q)>0$, we can give the voters who split $\mathcal{H}$ weight proportional to $J(p)$ and the other voters weight proportional to $\delta_p(q)$, so that the expected score of $\mathcal{H}$ is zero. Thus, $D$ is $g$-stable, with
\[
\frac{\mathbb{E}_v[L_v]}{\mathbb{E}_v[R_v]}
=q+\frac{p\,\delta_p(q)}{J(p)}.
\]
Requiring this ratio to be at most $\lambda$ forces
\begin{equation}\label{eq:binary}
    p\left[J(1-p+q)-J(q)\right]
    \leq(\lambda-q)J(p),
    \qquad\text{for }0\leq q\leq\min\{p,\lambda\}.
\end{equation}

\paragraph{Election 3. Three classes of voters: linking $\phi$ and $J$.}
Our first construction constrained $\phi$ pointwise by placing $a^\star$ outside $\mathrm{supp}(D)$, while our second constrained $J$ by balancing two classes in $\mathrm{supp}(D)$. To make $\phi$ and $J$ interact, we now combine the two ideas. We partition the support of $D$ into two classes, $\mathcal{H}$ (recycling the notation, this still means high values of $\rho$) with $\rho=2$ and probability $p$, and $\mathcal{M}$ (\underline{m}iddle values of $\rho$) with $\rho=1$ and probability $1-p$, while keeping $a^\star$ outside the support of $D$.

Choose $x,y,z$ with $0\leq x\leq 1-p$, $p\leq y<1$, and $p<z\leq 1$. We consider three types of voters, who place the block $\mathcal{H}_p$ at the top of their ranking, on the intermediate interval $[z-p,z]$, and at the bottom, respectively, while placing $a^\star$ at ranks $x$, $z$, and $y$:
\[
\begin{array}{c@{\quad}cc@{\quad}cc}
    \toprule
    \text{Preference order (worst to best)}&\text{Score of }\mathcal{H}&\text{Score of }a^\star&L_v&R_v\\
    \midrule
    \mathcal{M}_x,\ a^\star,\ \mathcal{M}_{1-p-x},\ \mathcal{H}_p
    &J(1-p)&-\phi(x)&x&2\\
    \mathcal{M}_{z-p},\ \mathcal{H}_p,\ a^\star,\ \mathcal{M}_{1-z}
    &J(z-p)-J(z)&-\phi(z)&z&1\\
    \mathcal{H}_p,\ \mathcal{M}_{y-p},\ a^\star,\ \mathcal{M}_{1-y}
    &-J(p)&-\phi(y)&p+y&1\\
    \bottomrule
\end{array}
\]
To read off $L_v$ and $R_v$, note that $a^\star$ splits each voter's rank line in two. Every candidate ranked above $a^\star$ has $\lowrho_v(b)=\rho(a^\star)=0$, so it contributes nothing to $L_v$ and its full $\rho$-value to the gap $\rho(b)-\lowrho_v(b)$. Below $a^\star$, $\lowrho_v$ is the running minimum of $\rho$.
\begin{itemize}
    \item In the first type, only $\mathcal{M}_x$ lies below $a^\star$, so $L_v=x$. Above $a^\star$, the block $\mathcal{H}_p$ has $\rho=2$, so $R_v=2$.
    \item In the second type, both $\mathcal{M}_{z-p}$ and $\mathcal{H}_p$ lie below $a^\star$, and the running minimum stays at $1$ across both, so $L_v=z$. The gap is $1$ on $\mathcal{H}_p$ and on $\mathcal{M}_{1-z}$ above $a^\star$, so $R_v=1$.
    \item In the third type, the running minimum is $2$ on $\mathcal{H}_p$ and $1$ on $\mathcal{M}_{y-p}$, so $L_v=2p+(y-p)=p+y$. Since the running minimum equals $\rho$ on both blocks, the only positive gap comes from $\mathcal{M}_{1-y}$ above $a^\star$, so $R_v=1$.
\end{itemize}

Any mixture of these three voter types makes $D$ $g$-stable whenever the expected score of $\mathcal{H}$ is zero and the expected score of $a^\star$ is at most zero, and every such mixture must satisfy $\mathbb{E}_v[L_v-\lambda R_v]\leq 0$. Choosing the voter weights to balance both scores yields an upper bound on $\phi(z)=J'(z)$ at a large rank $z$ in terms of $J(z)$ and the values of $\phi$ and $J$ at smaller ranks. Although we do not need this differential inequality in our formal proofs, it directly shapes $g^\star$ at large ranks, and the intermediate block configuration of the second voter type will reappear when we bound block scores in \Cref{sec:upper-g}.

We illustrate the activation $g^\star$ behind our upper bound and show which parts of it are shaped by Elections 1-3 in \Cref{fig:certificate}. Now, we turn to proving the lower and upper bounds.

\begin{figure}[htbp]
	\centering
	\begin{tikzpicture}[font=\small]
		\begin{axis}[axis lines=box,width=0.62\linewidth,height=0.42\linewidth,xlabel={rank $r$},ylabel={$g^\star(r)$},ylabel style={rotate=-90},ymin=0,ymax=1,xmin=0,xmax=1,xtick={0,0.5623,1},xticklabels={$0$,$\lambda$,$1$},ytick={0,0.30452,1},yticklabels={$0$,$\tau$,$1$},axis on top,axis line style={line width=.65pt},tick style={line width=.6pt},tick label style={font=\small},label style={font=\small}]
			\addplot[ybar interval,fill=green!10,draw=none] table[x=x,y=y] {figures/2.124/certificate-piece2.dat};
			\addplot[ybar interval,fill=blue!8,draw=none] table[x=x,y=y] {figures/2.124/certificate-piece4.dat};
			\addplot[ybar interval,fill=orange!12,draw=none] table[x=x,y=y] {figures/2.124/certificate-piece5.dat};
			\draw[gray,dash pattern=on 3pt off 2.5pt,line width=.5pt] (axis cs:0,0.30452) -- (axis cs:1,0.30452);
			\draw[gray,densely dotted,line width=.5pt] (axis cs:0.5623,0) -- (axis cs:0.5623,0.30452);
			\addplot[gray!70!black,dotted,line width=.8pt] coordinates {(0.11817,0) (0.25,0.09664) (0.33,0.11861) (0.39,0.15578) (0.74,0.45794) (0.84,0.51316) (1,0.60886)};
			\addplot[red!70!black,line width=1.2pt] table[x=x,y=g] {figures/2.124/certificate-g-stretched.dat};
			\addplot[only marks,mark=none,nodes near coords,point meta=explicit symbolic,every node near coord/.append style={anchor=center,font=\small}] table[x=x,y=y,meta=label] {figures/2.124/certificate-phase-labels.dat};
		\end{axis}
	\end{tikzpicture}
	\caption{A visualization of the activation $g^\star$ behind the upper bound in \Cref{thm:optimal-g}, with score function $\phi^\star$ and $J^\star(r)=\int_0^r\phi^\star$, where the roles of the three elections are color-coded. In region 1, $g^\star$ runs along the sloped dotted line, which encodes $\phi\leq\lambda-r$ in \Cref{eq:ceiling}. In region 2, at small ranks, $g^\star$ stays as low as Election 2 allows, which leaves the most room for \Cref{eq:binary} at large ranks. In region 3, $g^\star$ follows the differential inequality of Election 3. Near $r=1$, $g^\star$ rises steeply to $g^\star(1)=1$, far above the sloped line: voters who rank $a^\star$ first add to $L_v$ without adding to $R_v$ (Election 1), so the activation at the top must be large compared with its mean $\tau$.}
	\label{fig:certificate}
\end{figure}

\subsection{The lower bound}
\label{sec:lower-g}

Collecting the constraints from our first two election constructions, we define the set of \emph{admissible} score functions as follows.
\begin{definition}\label{def:admissible}
	For $\lambda\in(0,1)$, let $\mathcal A_\lambda$ be the set of continuous, nonincreasing functions $\phi:[0,1]\to\mathbb R$ with $\int_0^1\phi(r)\,dr=0$ that satisfy \Cref{eq:ceiling} and \Cref{eq:binary}.
\end{definition}
\begin{lemma}\label{lem:necessary}
	If $\operatorname{distortion}(g)\leq1+2\lambda$, then $c(\tau-g)\in\mathcal A_\lambda$ for some $c>0$.
\end{lemma}
\begin{proof}
	This follows directly from the necessary conditions derived in Elections 1 and 2.
\end{proof}

By \Cref{lem:necessary}, proving the lower bound in \Cref{thm:optimal-g} reduces to showing that $\mathcal A_\lambda$ is empty for $\lambda < 0.5616$. We test this by discretizing $[0,1]$ into a uniform grid. For $\lambda\in(0,1)$ and an integer $n\geq2$, set $r_i=i/n$ and $h=1/n$, and let $S(\lambda,n)$ denote the following linear system in the unknowns $\phi_0,\dots,\phi_{n-1}$ and $J_0,\dots,J_n$:
\begin{enumerate}[label=(\roman*)]
	\item $\phi_i\geq\phi_{i+1}$ for $i<n-1$, and $J_0=J_n=0$;
	\item $h\phi_{i+1}\leq J_{i+1}-J_i\leq h\phi_i$ for every $i$ (the left inequality only for $i<n-1$);
	\item $\phi_i\leq\lambda-r_i$ for every $i$;
	\item $(\lambda-r_b)\,J_a\geq r_a\,(J_{n-a+b}-J_b)$ for $1\leq a<n$ and $0\leq b\leq\min\{a,\lambda n\}$, which is \Cref{eq:binary} at $p=r_a$, $q=r_b$;
	\item\label{it:S-bounds} $-\lambda^2/(2(1-r_i))\leq\phi_i\leq\lambda$ and $0\leq J_i\leq\lambda^2/2$ for every $i$.
\end{enumerate}

\begin{lemma}\label{lem:discretization}
	If $\phi\in\mathcal A_\lambda$, then $\phi_i=\phi(r_i)$, $J_i=J(r_i)$ solves $S(\lambda,n)$.
\end{lemma}
\begin{proof}
	Condition (i) holds because $\phi$ is nonincreasing and $J(0)=0=\int_0^1\phi(r)\,dr=J(1)$. Condition (ii) follows from the fact that $J_{i+1}-J_i = \int_{r_i}^{r_{i+1}} \phi(u)\,du$ is $h$ times the average of $\phi$ on $[r_i,r_{i+1}]$, which lies between its endpoint values $\phi(r_{i+1})$ and $\phi(r_i)$. Conditions (iii) and (iv) simply evaluate \Cref{eq:ceiling} and \Cref{eq:binary} at the grid points. Finally, for condition (v), since $\phi$ is nonincreasing, $J$ is concave with $J(0)=J(1)=0$, which implies $J(r)\geq 0$ and $\phi(r)\geq(J(1)-J(r))/(1-r)=-J(r)/(1-r)$. Combined with the upper bound $J(r)\leq\int_0^{\min\{r,\lambda\}}(\lambda-u)\,du\leq\lambda^2/2$ from \Cref{eq:ceiling}, this yields the stated bounds on $\phi_i$ and $J_i$.
\end{proof}

\begin{proposition}[Lower Bound]\label{prop:lower-certificate}
	$\inf_g \operatorname{distortion}(g)\geq2.1232$.
\end{proposition}
\begin{proof}
	By \Cref{lem:necessary,lem:discretization}, it suffices to show that $S(\lambda,n)$ is infeasible for $\lambda=0.56162$ and some $n$. The file \href{https://github.com/YQL-Skorpion/stable-lotteries/blob/main/section_6_checker/lower_certificate.json}{\texttt{checker/lower\_certificate.json}} gives such an $n$ together with a certificate of infeasibility, which we check with \href{https://github.com/YQL-Skorpion/stable-lotteries/blob/main/section_6_checker/verify_lower_weights.py}{\texttt{checker/verify\_lower\_weights.py}} (\Cref{app:verification}). %
\end{proof}

\subsection{The upper bound}
\label{sec:upper-g}

We now turn the stability of $D$ into an upper bound on its distortion. Throughout this section, fix $\lambda\in(0,1)$ and let $\phi:[0,1]\to\mathbb R$ be a function satisfying the basic score-function properties from \Cref{sec:three-elections}, normalized so that $\phi(1)=-1$:
\begin{equation}\label{eq:score-shape}
	\begin{gathered}
		\phi\text{ is continuous and nonincreasing},\qquad\int_0^1\phi=0,\qquad\phi(1)=-1,\\
		\phi>0\text{ on }[0,\lambda),\qquad\phi(r)\leq\lambda-r\text{ on }[0,1].
	\end{gathered}
\end{equation}
In particular, these conditions imply that $\phi(\lambda)=0$ and $\phi(r)\leq0$ on $[\lambda,1]$. Reversing \Cref{eq:score}, the function
\[
	g_\phi(r)=\frac{\phi(0)-\phi(r)}{1+\phi(0)}
\]
is continuous and nondecreasing with $g_\phi(0)=0$ and $g_\phi(1)=1$, and satisfies $\phi=(1+\phi(0))(\tau-g_\phi)$ where $\tau=\int_0^1g_\phi(r)\,dr$. Thus $-\phi$ is a score function for $g_\phi$.

\subsubsection{Proving the upper bound}
\label{sec:upper-proof}

Given any election and a $g_\phi$-stable $D$, we wish to establish \Cref{eq:height-goal}. Just as in \Cref{sec:stable-k}, this inequality cannot hold for every voter $v$ individually, so we use the stability of $D$ to offset each voter's contribution by terms that average to (at least) zero across voters (analogous to the potential difference in \Cref{eq:key-lem-sl}). Here, we first slice $L_v$ by height. In \Cref{fig:rank-graph}, $L_v$ can be viewed as the shaded area under the running minimum $\lowrho_v$, so with the layer-cake formula, we can write it as an integral, over thresholds $t\geq 0$, of the width of the shaded region at height $t$. We offset each slice separately, and then introduce a potential function $W$ on ranks to telescope the resulting per-threshold bounds across $t$.

\paragraph{Step 1. Slicing and telescoping.}
For each threshold $t\geq 0$, let $S_t=\{b\in C:\rho(b)>t\}$, and let $p_t=D(S_t)$. On $v$'s rank line, a candidate $b$ satisfies $\lowrho_v(b)>t$ if and only if $b$ and every candidate ranked below $b$ belong to $S_t$. Thus, the candidates with $\lowrho_v(b)>t$ occupy an initial interval $[0,q_v(t))$, where $q_v(t)=\min_{c\notin S_t}\mathrm{rank}_v(c;D)$ is the rank of the first candidate outside $S_t$. Let $z_v(t)$ denote the rightmost endpoint of any candidate in $S_t$ with positive mass in $D$ (or $0$ if there is none). Finally, let $s_v=\Pr_{b\sim D}[a^\star\succ_v b]$ be the left endpoint of $a^\star$'s interval on $v$'s rank line. Since $\rho(a^\star)=0$, we have $a^\star\notin S_t$, and therefore $q_v(t)\leq s_v$ for all $t\geq 0$. Observe that $q_v(t)$ and $z_v(t)$ satisfy
\begin{equation}\label{eq:geometry}
	\int_0^\infty q_v(t)\,dt=L_v,
	\qquad
	z_v(t+R_v)\leq q_v(t).
\end{equation}
Indeed, the first relation is the layer-cake identity $\mathbb{E}_{b\sim D}[\lowrho_v(b)]=\int_0^\infty \Pr_{b\sim D}[\lowrho_v(b)>t]\,dt$. For the second, if some candidate $a\in S_{t+R_v}$ were ranked by $v$ above some candidate $c\notin S_t$, then we would have $\rho(a)>t+R_v$ and $\lowrho_v(a)\leq\rho(c)\leq t$, contradicting $R_v=\max_b(\rho(b)-\lowrho_v(b))$. Thus every candidate in $S_{t+R_v}$ precedes the first candidate outside $S_t$, which gives $z_v(t+R_v)\leq q_v(t)$ (see \Cref{fig:threshold-geometry}).

\begin{figure}[htbp]
	\centering
	\includegraphics[width=0.95\linewidth]{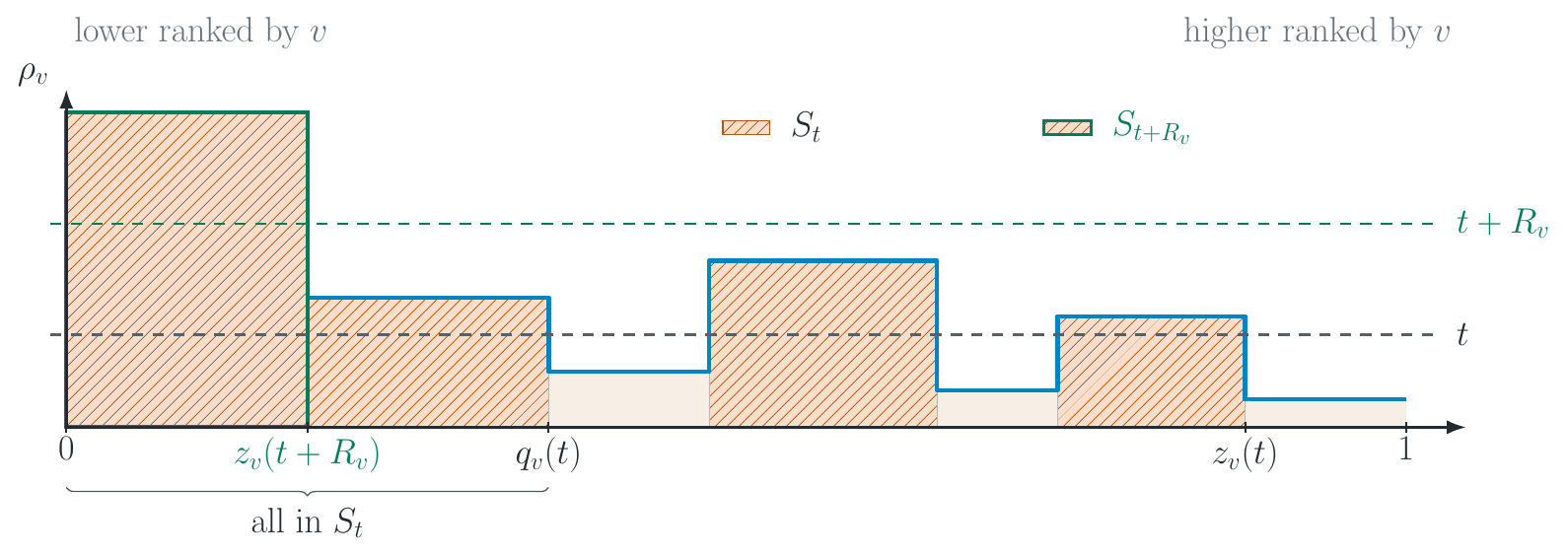}
	\caption{A visualization of $S_t$ and $S_{t+R_v}$ on a voter's rank line.}
	\label{fig:threshold-geometry}
\end{figure}

To bound $\int_0^\infty q_v(t)\,dt$, suppose we can modify $q_v(t)$ into some other $\widetilde{q}_v(t)$ satisfying $\mathbb{E}_v[q_v(t)]\leq\mathbb{E}_v[\widetilde{q}_v(t)]$ for every $t\geq0$, and find a nondecreasing function $W:[0,1]\to\mathbb{R}_{\geq 0}$ such that for every $v$ and $t\geq 0$,
\begin{equation}\label{eq:transfer}
	\begin{aligned}
		\widetilde{q}_v(t)&\leq\lambda-W(q_v(t)),\\
		\widetilde{q}_v(t)&\leq W(z)-W(q_v(t))
		&&\text{if }z_v(t)\leq z\leq s_v.
	\end{aligned}
\end{equation}
Then \Cref{eq:height-goal} follows. Indeed, for $0\leq t<R_v$ we use the first inequality in \Cref{eq:transfer}. For $t\geq R_v$, \Cref{eq:geometry} implies $z_v(t)\leq q_v(t-R_v)\leq s_v$, so the second inequality in \Cref{eq:transfer} with $z=q_v(t-R_v)$ gives $\widetilde{q}_v(t)\leq W(q_v(t-R_v))-W(q_v(t))$. When we integrate over $t\geq 0$, the $W$ terms telescope between thresholds separated by $R_v$, and since $W\geq0$, we obtain with $T = \max_b \rho(b)$,
\begin{align*}
\int_0^T \widetilde{q}_v(t)\,dt
&\leq \int_0^{R_v}
    \bigl(\lambda-W(q_v(t))\bigr)\,dt
    +\int_{R_v}^T
    \bigl(W(q_v(t-R_v))-W(q_v(t))\bigr)\,dt\\
&=\lambda R_v
    -\int_0^T W(q_v(t))\,dt
    +\int_0^{T-R_v}W(q_v(t))\,dt\\
&=\lambda R_v
    -\int_{T-R_v}^T W(q_v(t))\,dt\\
&\leq\lambda R_v.
\end{align*}
and so
\begin{equation}\label{eq:telescope}
	\int_0^\infty \widetilde{q}_v(t)\,dt
	\leq\lambda R_v.
\end{equation}
Taking expectations over $v$ in \Cref{eq:geometry,eq:telescope} and using $\mathbb{E}_v[q_v(t)]\leq\mathbb{E}_v[\widetilde{q}_v(t)]$ yields \Cref{eq:height-goal}. Thus, our task reduces to constructing $\widetilde{q}_v(t)$ and $W$ that satisfy \Cref{eq:transfer} and $\mathbb{E}_v[q_v(t)]\leq\mathbb{E}_v[\widetilde{q}_v(t)]$.

\paragraph{Step 2. Using stability to offset each slice.}
To build the offset in $\widetilde{q}_v(t)$, we use two consequences of the stability of $D$.\footnote{These two consequences, \Cref{eq:zero-score-blocks,eq:zero-score} below, are the only places where the stability of $D$ enters the proof. In particular, the conditions in the Verifier Lemma (\Cref{lem:verifier}) make no reference to the stability of $D$, or to any election at all.} First, since every candidate in $\mathrm{supp}(D)$ has expected score zero across voters by stationarity (\Cref{prop:gsl-stationary}), linearity of expectation shows that the total score $k_v(t)$ of $S_t$ on $v$'s rank line satisfies
\begin{equation}\label{eq:zero-score-blocks}
	\mathbb{E}_v[k_v(t)]=0.
\end{equation}
Second, since $\phi$ is nonincreasing, its value at the left endpoint $s_v$ of $a^\star$'s interval is at least its average over that interval; as $a^\star$ has expected score at most zero (\Cref{thm:gsl-key-property}), this gives
\begin{equation}\label{eq:zero-score}
	\mathbb{E}_v[\phi(s_v)]\geq0.
\end{equation}

How should we combine $k_v(t)$ and $\phi(s_v)$ to form $\widetilde{q}_v(t)$? Recall from Step 1 that $\widetilde{q}_v(t)$ must satisfy \Cref{eq:transfer} for every voter while keeping $\mathbb{E}_v[\widetilde{q}_v(t)]\geq\mathbb{E}_v[q_v(t)]$. So we want the offset to be negative for the voters whose slice width $q_v(t)$ is too large for \Cref{eq:transfer}, and positive only for voters with room to spare. Two kinds of voters need such an offset, and each stability fact handles one kind:
\begin{itemize}
    \item If $S_t$ fills the bottom of $v$'s rank line, then $q_v(t)$ is large, but the score $k_v(t)$ of $S_t$ is negative, as for $\mathcal{H}$ in Election 2. Adding a multiple of $k_v(t)$ helps these voters and, by \Cref{eq:zero-score-blocks}, leaves the average unchanged.
    \item By \Cref{eq:score-shape}, $\phi(s_v)>0$ when $s_v<\lambda$ and $\phi(s_v)\leq 0$ when $s_v\geq\lambda$. The voters with $s_v\geq\lambda$, such as the voters in Election 1 who rank $a^\star$ at the very top, are the ones for whom \Cref{eq:transfer} is hardest to satisfy. In contrast, the voters with $s_v<\lambda$ have $q_v(t)\leq s_v<\lambda$, leaving room in the first inequality of \Cref{eq:transfer}. By \Cref{eq:zero-score}, we can lower $\widetilde{q}_v(t)$ for the former in proportion to $-\phi(s_v)$, and compensate by raising $\widetilde{q}_v(t)$ for the latter.
\end{itemize}

Specifically, let $A=\mathbb{E}_v[\mathbf{1}_{\{s_v<\lambda\}}\phi(s_v)]$ and assume for now that $A>0$ and $0<p_t<1$. Given $\beta(p)\geq0$ and $c(p,x)\geq0$ for $p\in(0,1)$ and $x<\lambda$, we define
\begin{equation}\label{eq:corrected}
	\widetilde{q}_v(t)=q_v(t)+\beta(p_t)k_v(t)+
	\begin{cases}
		c(p_t,s_v),&s_v<\lambda,\\
		\phi(s_v)\,\bar c(p_t),&s_v\geq\lambda,
	\end{cases}
	\qquad
	\bar c(p)=\frac{\mathbb{E}_v[\mathbf1_{\{s_v<\lambda\}}c(p,s_v)]}{A}.
\end{equation}
The scaling in $\bar c(p_t)$ ensures that the last term in \Cref{eq:corrected} has nonnegative expectation $\bar c(p_t)A+\bar c(p_t)\mathbb{E}_v[\mathbf{1}_{\{s_v\geq\lambda\}}\phi(s_v)]=\bar c(p_t)\mathbb{E}_v[\phi(s_v)]\geq0$ by \Cref{eq:zero-score}, while $\beta(p_t)k_v(t)$ has expectation zero by \Cref{eq:zero-score-blocks}. Consequently,
\begin{equation}\label{eq:mean-correct}
	\mathbb{E}_v[q_v(t)]\leq\mathbb{E}_v[\widetilde{q}_v(t)].
\end{equation}

\paragraph{Step 3. Choosing the potential.}
We define $W$ through a family of functions $W_x$, indexed by $x\in[0,\lambda)$, where $W_x$ corresponds to the voters with $s_v=x$. Since such a voter has $q_v(t)\leq s_v<\lambda$, setting $W_x(z)=0$ for $z\leq\lambda$ means $W$ does not enter their inequalities in \Cref{eq:transfer}. For $z>\lambda$, we choose $W_x(z)$ so that when $p_t=1$ (as in Election 1) and $c(1,x)=\lambda-x$, the first inequality in \Cref{eq:transfer} holds with equality at every $q>\lambda$. Specifically, we define
\begin{equation}\label{eq:potential}
	W_x(z)=
	\begin{cases}
		0,&z\leq\lambda,\\[1mm]
		\displaystyle\lambda-z-\frac{\lambda-x}{\phi(x)}\phi(z),&z>\lambda,
	\end{cases}
\end{equation}
and let $W$ be their average over the voters with $s_v<\lambda$, weighted by $\phi(s_v)$ (and normalized by $A$, just like $\bar c(p)$ in \Cref{eq:corrected}):
\begin{equation}\label{eq:averages}
	W(z)=\frac{\mathbb{E}_v[\mathbf1_{\{s_v<\lambda\}}\phi(s_v)W_{s_v}(z)]}{A}.
\end{equation}
Because $0<\phi(x)\leq\lambda-x$ by \Cref{eq:score-shape}, the line through $(x,\phi(x))$ and $(\lambda,0)$ has slope at least $-1$. Provided that $\phi(z)+z$ is nonincreasing on $[\lambda,1]$, $\phi$ falls at least as steeply as this line beyond $\lambda$, so the horizontal gap $W_x(z)$ between them in \Cref{fig:potential-geometry} is nonnegative and grows with $z$. Hence $W$ is nonnegative and nondecreasing.

\begin{figure}[htbp]
	\centering
	\includegraphics[width=0.75\linewidth]{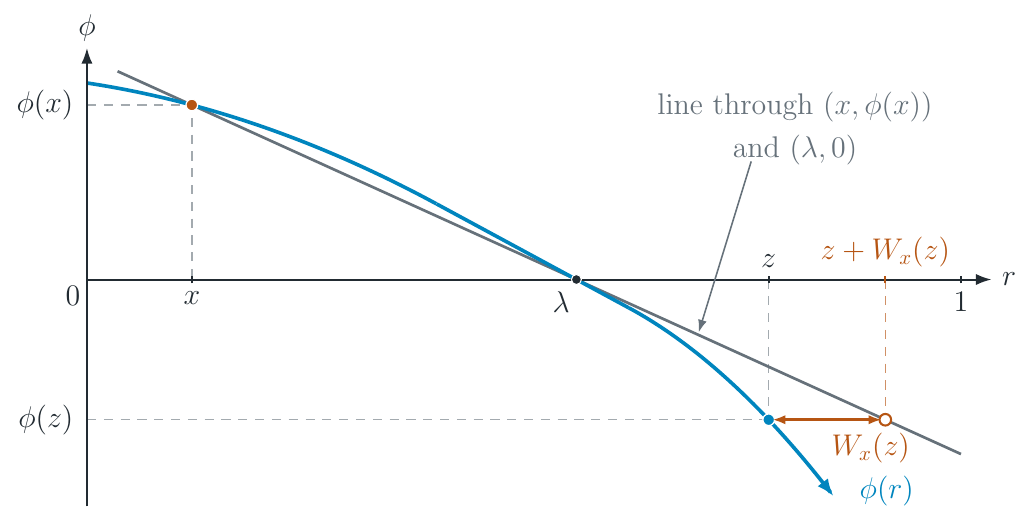}
	\caption{The potential $W_x$. Choosing $x<\lambda$ fixes the line through $(x,\phi(x))$ and $(\lambda,0)$. The orange horizontal gap between the two curves at height $\phi(z)$ equals $W_x(z)$.}
	\label{fig:potential-geometry}
\end{figure}

\paragraph{Step 4. Checking each voter.}
To bound $\widetilde{q}_v(t)$ for an individual voter, we first bound $k_v(t)$ in terms of $p=p_t$, $q=q_v(t)$, and $z\geq z_v(t)$. By definition, $S_t$ completely fills the initial interval $[0,q)$, which contributes $-J(q)$ to $k_v(t)$, while the remaining mass $p-q$ of $S_t$ lies inside $[q,z]$. Since $-\phi$ is nondecreasing, this remaining mass achieves its largest possible score when pushed as far to the right as allowed, so either to the top of the rank line on $[1-p+q,1]$, or against the upper bound $z$ on $[z-p+q,z]$:
\begin{equation}\label{eq:score-blocks}
	\begin{aligned}
		k_v(t)&\leq\delta_p(q)=J(1-p+q)-J(q),\\
		k_v(t)&\leq K_p(q,z)=J(z-p+q)-J(z)-J(q).
	\end{aligned}
\end{equation}
Observe that the first inequality encodes the score of $\mathcal{H}$ for the voters in Election 2 who split it around $\mathcal{Z}$ (hence the notation $\delta_p(q)$), and at $q=0$, the second inequality encodes the score of $\mathcal{H}$ for the second voter type in Election 3.

We can now check what \Cref{eq:transfer} requires voter by voter when $0<p_t<1$ and $A>0$, with $p$ and $q$ as above and $z\in[z_v(t),s_v]$.
\begin{itemize}
	\item \emph{For a voter with $s_v=x<\lambda$:} Here $q\leq x<\lambda$, so $W(q)=0$. Moreover, the second inequality in \Cref{eq:transfer} only applies when $z\leq x<\lambda$, where $W(z)=0$ and $K_p(q,z)\leq K_p(q,x)$ because $K_p(q,z)$ is nondecreasing in $z$. By \Cref{eq:score-blocks,eq:corrected}, \Cref{eq:transfer} holds for this voter whenever
	\begin{equation}\label{eq:low-cuts}
		\begin{aligned}
			q+\beta(p)\delta_p(q)+c(p,x)&\leq\lambda
			&&\text{for }0\leq q\leq\min\{p,x\},\\
			q+\beta(p)K_p(q,x)+c(p,x)&\leq0
			&&\text{for }0\leq q\leq p\leq x.
		\end{aligned}
	\end{equation}
	\item \emph{For a voter with $s_v\geq\lambda$:} Here $q\leq s_v$, and $z\leq s_v$ whenever the second inequality in \Cref{eq:transfer} applies. Since $\phi$ is nonincreasing and $\phi(s_v)\leq 0$, both $(-\phi(q))_+$ and $(-\phi(z))_+$ are at most $-\phi(s_v)$. Thus, by \Cref{eq:score-blocks}, \Cref{eq:transfer} holds for this voter whenever
	\begin{gather*}
	q+\beta(p)\delta_p(q)+W(q)-\lambda \leq (-\phi(q))_+\bar{c}(p),\\
	q+\beta(p)K_p(q,z)+W(q)-W(z) \leq (-\phi(z))_+\bar{c}(p).
	\end{gather*}
	These two inequalities follow by averaging over the voters with $s_v<\lambda$, provided that for every $x\in[0,\lambda)$,
	\begin{equation}\label{eq:high-cuts}
		\begin{aligned}
			\phi(x)\left[q+\beta(p)\delta_p(q)+W_x(q)-\lambda\right]
			&\leq(-\phi(q))_+c(p,x)
			&&\text{for }0\leq q\leq p,\\
			\phi(x)\left[q+\beta(p)K_p(q,z)+W_x(q)-W_x(z)\right]
			&\leq(-\phi(z))_+c(p,x)
			&&\text{for }0\leq q\leq p\leq z\leq1.
		\end{aligned}
	\end{equation}
	Indeed, after dividing by $\phi(x)>0$, each inequality in \Cref{eq:high-cuts} is linear in $c(p,x)/\phi(x)$ and $W_x$, with coefficients that do not depend on $x$, and averaging with weights $\phi(x)/A$ turns these into $\bar c(p)$ and $W$, by \Cref{eq:corrected,eq:averages}.
\end{itemize}

The remaining cases ($p_t=0$, $p_t=1$, and $A=0$) are handled without any new conditions:
\begin{itemize}
	\item If $p_t=0$, we set $\widetilde{q}_v(t)=0$; then $q_v(t)=z_v(t)=0$, so \Cref{eq:mean-correct,eq:transfer} hold trivially.
	\item If $p_t=1$, we use \Cref{eq:corrected} with $\beta(1)=0$ and $c(1,x)=\lambda-x$. Then \Cref{eq:low-cuts} reduces to $q\leq x$, which holds because $q\leq s_v=x$. When $q\leq\lambda$, the first inequality in \Cref{eq:high-cuts} reduces to $\phi(x)(q-\lambda)\leq0$. When $q>\lambda$, it holds with equality by our choice of $W_x$ in \Cref{eq:potential}. For the second inequality in \Cref{eq:high-cuts}, $p_t=1$ forces $z=1$. Since $W_x(1)=\lambda-1+(\lambda-x)/\phi(x)\geq\lambda$ by \Cref{eq:score-shape}, the case $q\leq\lambda$ holds immediately, while for $q\geq\lambda$ the definition of $W_x$ and $-\phi(q)\leq1$ give
	\[
	\phi(x)\left[q+W_x(q)-W_x(1)\right]
	=\phi(x)\left[\lambda-W_x(1)\right]-(\lambda-x)\phi(q)
	\leq\lambda-x=(-\phi(1))_+c(1,x).
	\]
	\item Finally, if $A=0$, then \Cref{eq:zero-score} together with \Cref{eq:score-shape}, which gives $\phi>0$ on $[0,\lambda)$ and $\phi<0$ on $(\lambda,1]$, forces $s_v=\lambda$ for every voter. Setting $W=W_0$ and $\widetilde{q}_v(t)=q_v(t)+\beta(p_t)k_v(t)$ gives \Cref{eq:mean-correct} with equality by \Cref{eq:zero-score-blocks}, and since all ranks $q,z$ used in \Cref{eq:transfer} are at most $s_v=\lambda$, the terms $(-\phi(q))_+$ and $(-\phi(z))_+$ in \Cref{eq:high-cuts} vanish, so dividing \Cref{eq:high-cuts} at $x=0$ by $\phi(0)>0$ gives \Cref{eq:transfer}.
\end{itemize}

Together, Steps 1-4 show that whenever $\beta$ and $c$ satisfy \Cref{eq:low-cuts,eq:high-cuts}, the functions $\widetilde{q}_v(t)$ and $W$ satisfy \Cref{eq:transfer,eq:mean-correct} at every threshold, and hence \Cref{eq:height-goal} holds by Step 1. We summarize this in the following lemma, and then exhibit $\phi^\star$, $\beta$, and $c$ satisfying its conditions in \Cref{sec:upper-verify}.

\begin{lemma}[Verifier Lemma]\label{lem:verifier}
	Let $\phi$ satisfy \Cref{eq:score-shape} and let $\phi(r)+r$ be nonincreasing on $[\lambda,1]$. Let $\beta:(0,1)\to\mathbb R_{\geq0}$ and $c:(0,1)\times[0,\lambda)\to\mathbb R_{\geq0}$ satisfy \Cref{eq:low-cuts,eq:high-cuts} for every $p\in(0,1)$ and $x\in[0,\lambda)$. Then every $g_\phi$-stable lottery has distortion at most $1+2\lambda$.
\end{lemma}
\begin{proof}
	This follows from Steps 1-4, as explained above, together with \Cref{thm:biased-iff}.
\end{proof}

\subsubsection{Verifying the upper bound}
\label{sec:upper-verify}

To complete the upper bound in \Cref{thm:optimal-g}, we pair $\phi^\star$ from \Cref{fig:certificate} with $\beta$ and $c$ that satisfy the hypotheses of \Cref{lem:verifier}.

\begin{lemma}[Upper Bound]\label{lem:certificate}
	The functions $\phi^\star$, $\beta$, and $c$ recorded in \href{https://github.com/YQL-Skorpion/stable-lotteries/blob/main/section_6_checker/certificate.json}{\texttt{checker/}\allowbreak\texttt{certificate.json}} satisfy the hypotheses of \Cref{lem:verifier}, with $1+2\lambda\leq2.1248$.
\end{lemma}
\begin{proof}
	We check the claims with \href{https://github.com/YQL-Skorpion/stable-lotteries/blob/main/section_6_checker/verify.py}{\texttt{checker/verify.py}}. For more details, see \Cref{app:optimal-g}.
\end{proof}

\begin{proof}[Proof of \Cref{thm:optimal-g}]
	Immediate from \Cref{prop:lower-certificate} and \Cref{lem:verifier,lem:certificate}.
\end{proof}

\section{Discussion}\label{sec:discussion}

Perhaps the most striking aspect of this work is how often the same family of voting rules turns out to be useful. Stable lotteries and their relatives have been central to strong results for utilitarian distortion \citep{DBLP:journals/teco/EbadianKPS24} and metric distortion \citep{DBLP:journals/jacm/CharikarRWW24,frank2026improved,ye2026half}, as well as for committee selection problems involving proportional representation \citep{DBLP:journals/teco/ChengJMW20,jiang2020approximately}, Condorcet winning sets \citep{DBLP:conf/stoc/CharikarLRV025,DBLP:conf/soda/SongNL26}, and approximately dominating sets \citep{DBLP:conf/soda/CharikarR026}.  Notably, some of these applications go beyond randomized voting by using the probabilistic method to prove results about purely deterministic objectives. As \citet{charikar2026game} discuss, the success of these ideas across settings suggests that they capture something more general about what makes a collective choice reasonable. Our result supports this narrative. A family developed to study representation and stability also comes remarkably close to optimizing a rather different notion of social efficiency. For any new voting problem, it now seems worth asking early on what a suitable $g$-stable lottery can tell us.

\paragraph{Closing the remaining gap.} The obvious remaining question is to close the gap between $2.11264$ and $2.1248$. The current lower bound comes from such a particular and unnatural election that, if it is tight, we find it hard to believe that a natural voting rule would match it exactly. From this perspective, it is frankly surprising that the natural family of $g$-stable lotteries, developed for entirely different questions, can already get so close.

\paragraph{Other directions.} Now that we are near the end of the road towards optimal randomized metric distortion, it is worth remembering that the real treasure was the voting rules we made along the way. And since the goal is really to develop a deeper understanding of voting as a whole, the metric distortion framework still has a lot to offer. For example, breaking the barrier of $3$ opened up questions about doing so with less randomness \citep{DBLP:journals/corr/abs-2602-08871,jia2026robust} 
or less preference information \citep{DBLP:conf/sigecom/CharikarRT025,chan2026improved}. These kinds of questions still remain exciting avenues for future work, and the improvements in this paper only open up more space in these directions.

\subsection*{AI Methodology}

Our starting point was the idea to understand the metric distortion of stable $k$-lotteries. GPT-5.6 Sol was able to prove a distortion bound that surprised us ($2.333\ldots$ for $k = 2,3$ and generally \Cref{thm:dist-sl}), including the nice ideas of doing the analysis voter-by-voter and using stationarity together with the carefully constructed potential functions. We then asked it to generalize to $g$-stable lotteries and optimize $g$ for distortion, and it was able to prove a distortion bound initially of $2.236\ldots$ (with a proof resembling that of \Cref{thm:dist-sl}), and then with more prodding, $2.137\ldots$ with a substantially more involved proof. 

Over several weeks, we worked with the model to digest and simplify the proof, and improve the final computer search, all of which occasionally gave marginal improvements in the result. We also developed more intuitive ways of explaining the ideas in the proof (at least for humans), including visualizations described in the paper. The first version of the paper gave an upper bound of $2.13713$, and  GPT-6 Astra played a minor role in improving the computer search in that argument.

We then set GPT-6 Astra on pinning down the exact optimal distortion of $g$-stable lotteries. It was eventually able to give tight upper and lower bounds around $2.12466$, using an intricate variational argument and extensive casework, and we used Gemini 4 Argon to put together a write-up of this result (without human authorship), and GPT-6 Astra to give a formalized Lean proof within $10^{-5}$. Both are available in the \href{https://github.com/YQL-Skorpion/stable-lotteries}{GitHub repository} cited in \Cref{sec:optimal-g}. We worked with several models, namely GPT-6 Astra, Claude Opus 5.5, and Gemini 4 Argon, to cut away as much of the complexity as possible with minimal impact on the bounds, giving \Cref{thm:optimal-g}. The final paper is written by the authors, with assistance from GPT-6 Astra and Claude Opus 5.5 in copyediting and producing figures.

\subsection*{Acknowledgments}
Moses Charikar was a coauthor on an earlier version of the paper, and we thank him for many valuable discussions on this work.

Part of this work was done while PR was supported by the Stanford Input-Output Global Research Hub. JH is supported by the Simons Foundation Collaboration on the Theory of Algorithmic Fairness and the Simons Foundation Investigators Award 17351.

\clearpage
\bibliographystyle{references}
\bibliography{references}

\clearpage
\appendix
\crefalias{section}{appendix}

\section{A matching lower bound for stable \texorpdfstring{$k$}{k}-lotteries}
\label{app:lower-bound-stable-k}

We show that the distortion bound in \Cref{thm:dist-sl} is tight for every integer $k > 1$. We use elections that are the same as Election 1 from \Cref{sec:optimal-g}, with one optimal candidate $a^\star$ and a set $U$ of symmetric candidates. A uniform lottery over $U$ is stable, but has high social cost in a consistent metric.

\begin{theorem}\label{thm:lower-bound-stable-k}
For every integer $k > 1$, the supremum of the metric distortion over all elections and their stable $k$-lotteries is $1 + 2\lambda_k$ where
\[
    \lambda_k = \frac{1 + (k-1)(k+1)^{-1/(k-1)}}{k}.
\]
\end{theorem}

Fix $k > 1$ and let $\alpha=(k+1)^{-1/(k-1)}$. For any integer $m\geq\lceil 1/\alpha\rceil$, set
\[
    t=\lfloor m\alpha\rfloor,\qquad z=\frac{t}{m}.
\]
Then $1\leq t<m$ and $0<z\leq\alpha$. Let $C = \{a^\star\}\cup U$ where $|U|=m$, and define
\begin{equation}\label{eq:lower-bound-stable-k-r}
    r = \frac{k}{(k+1)(1-z^k)}.
\end{equation}
Since $z^k \leq \alpha^k = \alpha/(k+1) < 1/(k+1)$, we have $0 < r < 1$.

We partition the voters into two groups $V_0$ and $V_1$, comprising fractions $1-r$ and $r$ of the electorate, respectively. In each group, every permutation $u_1,\dots,u_m$ of $U$ occurs equally often. The corresponding preferences are
\[
\begin{array}{ll}
    V_0: & a^\star \succ_v u_1 \succ_v \cdots \succ_v u_m,\\[3pt]
    V_1: & u_1 \succ_v \cdots \succ_v u_{m-t} \succ_v a^\star
        \succ_v u_{m-t+1} \succ_v \cdots \succ_v u_m.
\end{array}
\]
Since $r$ is rational, these proportions can be realized by a finite electorate by taking an appropriate number of copies of each ranking. Let $D$ be the uniform distribution over $U$.

By symmetry, every candidate in $U$ has the same probability of being preferred by a random voter to all $k$ independent draws from $D$. Our choice of $r$ ensures that $a^\star$ has this same probability, which will establish stability.

\begin{lemma}\label{lem:lower-bound-stable-k-stability}
The lottery $D$ is a stable $k$-lottery in this election.
\end{lemma}

\begin{proof}
We first compute the probability that a random voter prefers a candidate $u\in U$ to all $k$ independent draws from $D$. The profile and $D$ are invariant under relabeling the candidates in $U$, so this probability is the same for every $u\in U$. If we also draw $u$ from $D$, each of the $k+1$ samples is equally likely to be the voter's favorite. Thus, the probability is $1/(k+1)$ for each $u\in U$.

It remains to compute the corresponding probability for $a^\star$. For voters in $V_0$, candidate $a^\star$ is preferred to every draw from $D$. For voters in $V_1$, it is preferred to all $k$ draws exactly when all of them belong to the $t$ candidates ranked below it. Hence,
\[
    \Pr_{\substack{v\sim V\\b_1,\dots,b_k\sim D}}
    [a^\star\succ_v b_1,\dots,b_k]
    = 1-r+rz^k
    = 1-r(1-z^k)
    = \frac{1}{k+1}.
\]
Thus every candidate is preferred by a random voter to all $k$ draws from $D$ with probability $1/(k+1)$. For any defending lottery, drawing the attacker from that same lottery gives a winning probability of $1/(k+1)$ by symmetry among the samples. The maximum over attackers is therefore at least $1/(k+1)$. Consequently, $D$ attains the minimum in \Cref{eq:minimax-sl} and is a stable $k$-lottery.
\end{proof}

We now use the biased metrics from \Cref{sec:biased-metrics} to show that $D$ has high distortion. It suffices to assign the same positive $\rho$ value to every candidate in $U$.

\begin{proof}[Proof of \Cref{thm:lower-bound-stable-k}]
By \Cref{lem:lower-bound-stable-k-stability}, the lottery $D$ is a stable $k$-lottery in the election constructed above. To bound its distortion from below, consider the biased metric from \Cref{def:biased} with
\[
    \rho(a^\star)=0,\qquad \rho(u)=2\quad\text{for every }u\in U.
\]
For a voter $v\in V_0$, we have $\lowrho_v(u)=2$ for every $u\in U$ and $\max_b(\rho(b)-\lowrho_v(b))=0$. For a voter $v\in V_1$, we have $\lowrho_v(u)=0$ for candidates ranked above $a^\star$ and $\lowrho_v(u)=2$ for candidates ranked below it, so $\max_b(\rho(b)-\lowrho_v(b))=2$. Since each voter in $V_1$ ranks a $z$ fraction of $U$ below $a^\star$, averaging over both groups gives
\[
    \Ev_{\substack{v\sim V\\b\sim D}}[\lowrho_v(b)]=2(1-r+rz),
    \qquad
    \Ev_{v\sim V}\left[\max_b(\rho(b)-\lowrho_v(b))\right]=2r.
\]
In this metric, voters in $V_0$ are colocated with $a^\star$. Each voter $v\in V_1$ is at distance $1$ from $a^\star$ and from the candidates in $T_v=\{u\in U:u\succ_v a^\star\}$, and at distance $3$ from the remaining candidates, as illustrated in \Cref{fig:lower-bound-stable-k-metric}.

\begin{figure}[ht]
    \centering
    \begin{tikzpicture}[x=1cm,y=1cm,
        candidate/.style={circle,fill=red!80!black,draw=red!45!black,
            inner sep=0pt,minimum size=6pt},
        voter/.style={circle,fill=RoyalBlue,draw=blue!55!black,
            inner sep=0pt,minimum size=6pt},
        every label/.style={font=\small},
        metriclabel/.style={fill=white,inner sep=2pt,font=\small}]
        \node[candidate,label=below:{\textcolor{red!80!black}{$a^\star$},\ \textcolor{RoyalBlue}{$V_0$}}] (a) at (0,0) {};
        \node[voter,label={[RoyalBlue]above:{$v\in V_1$}}] (v) at (0,2) {};
        \node[candidate,label={[red!80!black]above:{$u\in T_v$}}] (u) at (5,2) {};
        \node[candidate,label={[red!80!black]below:{$u'\in U\setminus T_v$}}] (b) at (5,0) {};
        \draw (a) -- node[metriclabel] {$1$} (v);
        \draw (v) -- node[metriclabel] {$1$} (u);
        \draw (a) -- node[metriclabel] {$2$} (b);
        \draw (a) -- node[metriclabel,pos=.3] {$2$} (u);
        \draw[dashed] (v) -- node[metriclabel,pos=.3] {$3$} (b);
        \begin{scope}[shift={(a.center)}]
            \filldraw[fill=red!80!black,draw=red!45!black]
                (-3pt,0) arc[start angle=180,end angle=0,radius=3pt] -- cycle;
            \filldraw[fill=RoyalBlue,draw=blue!55!black]
                (3pt,0) arc[start angle=0,end angle=-180,radius=3pt] -- cycle;
        \end{scope}
    \end{tikzpicture}
    \caption{Distances in the lower-bound construction for a voter $v\in V_1$ and two candidates on opposite sides of $a^\star$ in her ranking. Blue points represent voters and red points represent candidates; the split marker indicates that $V_0$ is colocated with $a^\star$.}
    \label{fig:lower-bound-stable-k-metric}
\end{figure}

By the social-cost identities in \Cref{sec:biased-metrics}, we have $\SC(a^\star)=r$ and $\Ev_{b\sim D}[\SC(b)-\SC(a^\star)]=2(1-r+rz)$. Moreover, $a^\star$ is optimal because $\lowrho_v(b)\geq 0$ for every voter $v$ and candidate $b$. Therefore,
\begin{align*}
    \frac{\Ev_{b\sim D}[\SC(b)]}{\SC(a^\star)}
    &=1+2\left(\frac{1}{r}-1+z\right)\\
    &=1+2\left(\frac{1}{k}+z-\frac{k+1}{k}z^k\right).
\end{align*}
As $m\to\infty$, we have $z\to\alpha$. Since $(k+1)\alpha^{k-1}=1$, the distortion converges to
\[
    1+2\left(\frac{1}{k}+\alpha-\frac{k+1}{k}\alpha^k\right)
    =1+2\left(\frac{1+(k-1)\alpha}{k}\right)
    =1+2\lambda_k.
\]
Together with the upper bound in \Cref{thm:dist-sl}, this proves the theorem.
\end{proof}

\section{$g$-Stable lotteries are optimal under small-tournament information}\label{app:small-tournament}

In this section, we prove the following lower bounds for small-tournament rules.

\begin{theorem}\label{thm:small-tournament-distortion}
The metric distortion of any $2$-, $3$-, or $4$-tournament rule is at least $3$, $7/3$, or $11/5$, respectively.
\end{theorem}

Consequently, by \Cref{thm:dist-sl}, stable $2$-lotteries are optimal among all $3$-tournament rules. Since maximal lotteries are known to have distortion exactly $3$ \citep*{DBLP:journals/jacm/CharikarRWW24}, they are likewise optimal among $2$-tournament rules. In fact, $g$-stable lotteries are optimal among $\ell$-tournament rules for every $\ell \leq 4$: using the same techniques as in \Cref{sec:stable-k}, one can show that $g$-stable lotteries with activation function $g(t) = t - \frac32 t^2 + \frac32 t^3$ have metric distortion exactly $11/5$.\footnote{If $g$ is a polynomial of degree less than $\ell$, then the $g$-stable lottery is an $\ell$-tournament rule: the payoff function of the $g$-activated Condorcet game can be written as a linear combination of the quantities $\Pr_{v\sim V}[a\succ_v b_1,\dots,b_j]$ with $j < \ell$.}

Suppose there are $m + 1$ candidates, and designate one candidate $a^\star$ as the optimal one. We call the remaining $m$ candidates \textit{ordinary}. The lower-bound construction consists of two key steps. First, to ``confuse'' the tournament rule, we ensure that $a^\star$ is indistinguishable from the other $m$ candidates. Specifically, for every set $S$ with $\abs{S} \leq \ell$, the distribution of the voters' orderings of $S$ is identical to that of a uniformly random permutation of $S$. In this case, no $\ell$-tournament rule can distinguish among the $m + 1$ candidates.

Second, since the rule cannot distinguish the candidates, our next goal is to make the optimal candidate $a^\star$ as much cheaper than every other candidate as possible. Intuitively, the more voters rank $a^\star$ first, the larger the gap.

\paragraph{Hiding the optimal candidate $a^\star$.}
We start with the case $\ell = 2$. Construct the preference profile so that each voter ranks the $m$ ordinary candidates in a uniformly random order and then places $a^\star$ either at the very top or at the very bottom, each with probability $1/2$. Under this sampling, for each pair of candidates, each one beats the other with probability $1/2$. Thus, no $2$-tournament rule can distinguish among the $m + 1$ candidates.

However, a $3$-tournament rule would notice that $a^\star$ is never in the middle of any triple containing it, so the same construction fails for $\ell \geq 3$. Instead, we must spread the position of $a^\star$ more evenly across each voter's ranking.

We control the position of $a^\star$ more carefully using a random variable. Fix a voter $v$. Rather than placing $a^\star$ at the top or the bottom of $v$'s ranking, voter $v$ draws a score $Y \in [0, 1]$ for $a^\star$ from some distribution. Meanwhile, each ordinary candidate receives an independent score drawn from $\Unif(0, 1)$. Voter $v$ then ranks the candidates in decreasing order of their scores. We determine the ranking of every voter $v$ by drawing these scores independently across voters.

Under this construction, the probability that a voter ranks $a^\star$ above $k$ given other candidates is exactly $\Ev[Y^k]$. It turns out that $a^\star$ is hidden when the first $\ell - 1$ moments of $Y$ match those of $\Unif(0, 1)$.

\begin{lemma}\label{lem:hiding_moments} Suppose that $m \geq \ell - 1$. In the election above, for every $k \leq \ell$ and every set of $k$ distinct candidates $\{c_1, c_2, \dots, c_k\}$, each of the $k!$ orderings occurs with probability $1/k!$ if
    \[
    \Ev[Y^j] = \frac1{j + 1} \qquad \text{for every }j = 1, \dots, \ell - 1.
    \]
\end{lemma}
\begin{proof}[Proof of \Cref{lem:hiding_moments}]
If all $k$ candidates are ordinary, then their scores are i.i.d.\@ draws from $\Unif(0, 1)$, so they are ordered uniformly at random and the claim follows immediately.

Otherwise, consider an ordering $c_1 \succ_v c_2 \succ_v \cdots \succ_v c_k$ with $c_b = a^\star$. The probability that this ordering occurs is exactly
\[
\Ev\left[\frac{Y^{(k - b)}}{(k - b)!}\cdot \frac{{(1 - Y)}^{(b - 1)}}{(b - 1)!}\right] = \frac{\Ev[Y^{(k - b)}{(1 - Y)}^{(b - 1)}]}{(k - b)!(b - 1)!}.
\] Expanding $(1 - Y)^{b - 1}$, this probability is a linear combination of $\Ev[Y^j]$ for $j \leq k - 1 \leq \ell - 1$. Therefore, as long as every such $\Ev[Y^j]$ equals its value under $Y \sim \Unif(0, 1)$ (i.e., $\Ev[Y^j] = \frac1{j + 1}$), the probability equals $1/k!$.
\end{proof}

Define the set of $Y$ that matches the first $\ell - 1$ moments by
\[
\mathcal{Y}_\ell=\left\{Y\in[0,1]\ :\ \Ev[Y^j]=\frac{1}{j+1}\ \text{ for every } j=1,\dots,\ell-1\right\}.
\]

\paragraph{Making $\SC(a^\star)$ cheap relative to other candidates.} Whenever a voter $v$ ranks $a^\star$ first, we place voter $v$ at $a^\star$. This forces $\rho(c)$ to take the same value for every ordinary candidate $c$: if there were a pair of ordinary candidates $(c, c^\prime)$ with $\rho(c) > \rho(c^\prime)$, then some voter $v$ who ranks $a^\star$ first would also rank $c \succ_v c^\prime$, making $d(v, a^\star)$ positive by \Cref{def:biased}. Without loss of generality, we consider the biased metric with\[
\rho(a^\star) = 0, \qquad \rho(c) = 1 \qquad \text{for every ordinary candidate $c$.}
\] The following lemma computes the distortion for choosing an ordinary candidate in this biased metric.

\begin{lemma}\label{lem:metric-distortion-ordinary} Suppose that $\Ev[Y] = 1/2$. For any ordinary candidate $c$, choosing $c$ in the aforementioned biased metric has distortion
    \[
    \frac{\SC(c)}{\SC(a^\star)} = 1 + \frac1{1 - \Ev[Y^m]}.
    \]
\end{lemma}

 As a result, the more voters putting $a^\star$ at the top (with probability $\Ev[Y^m]$), the larger the distortion of an incorrect choice $c$. We defer the proof of \Cref{lem:metric-distortion-ordinary} to the end of this section.

 \begin{corollary}\label{cor:metric-distortion-tournament}
    Let $\ell \geq 2$ and $m \geq \ell - 1$, and $Y \in \mathcal{Y}_\ell$. On elections with $m + 1$ candidates, every $\ell$-tournament rule has distortion at least
    \[
    1 + \frac{m}{m + 1} \cdot \frac1{1 - \Ev[Y^m]}.
    \]
 \end{corollary}

 \begin{proof}[Proof of \Cref{cor:metric-distortion-tournament}] Fix an $\ell$-tournament rule, and let $D$ be its output when every ordering of every set of $k\leq \ell$ candidates has equal frequency. Some candidate receives probability at most $1/(m + 1)$ under $D$. Since $a^\star$ is hidden from the view of the rule (\Cref{lem:hiding_moments}), we can choose that candidate to be $a^\star$. Therefore, by \Cref{lem:metric-distortion-ordinary} we have
    \[
    \frac{\Ev_{b \sim D}[\SC(b)]}{\SC(a^\star)} = 1 + (1 - D(a^\star)) \cdot \frac1{1 - \Ev[Y^m]} \geq 1 + \frac{m}{m + 1} \cdot \frac1{1 - \Ev[Y^m]}.\qedhere
    \]
 \end{proof}

 As $m$ tends to infinity, $\Ev[Y^m]$ tends to $\Pr[Y = 1]$. Finally, we can prove \Cref{thm:small-tournament-distortion} by resolving the following optimization.
\begin{equation}\label{eq:moment-opt}
\text{maximize}\quad \Pr[Y = 1]\qquad\text{subject to}\quad Y\in\mathcal{Y}_\ell.
\end{equation}

\begin{proof}[Proof of \Cref{thm:small-tournament-distortion}]
    We introduce an auxiliary polynomial $H$ which satisfies
\begin{equation}\label{eq:poly_H}
    \deg(H) < \ell, \qquad H(1) = 1, \qquad \text{and} \qquad  H(t) \geq 0 \text{ for every } t \in [0, 1].
\end{equation}
    For any $Y \in \mathcal{Y}_\ell$, since $Y$ matches the first $\ell-1$ moments of $\Unif(0,1)$, we have
\[
\int_{[0,1]}H(Y)\dif Y=\int_0^1 H(t)\dif t.
\]
On the other hand, since $H$ is nonnegative on $[0,1]$,
\[
\int_{[0,1]}H(Y)\dif Y=H(1)\Pr[Y=1]+\int_{[0,1)}H(Y)\dif Y\geq \Pr[Y=1].
\]
Therefore, $\Pr[Y=1]\leq\int_0^1 H(t)\dif t$, with equality if and only if $Y$ is supported on $1$ and the zeros of $H$ in $[0,1)$.

In other words, $H$ certifies an upper bound on $\Pr[Y=1]$. If we find a pair $(Y, H)$ such that
\begin{itemize}
\item $Y \in \mathcal{Y}_\ell$ and $Y$ is supported on $1$ and the zeros of $H$, and
\item $H$ satisfies the conditions in \Cref{eq:poly_H},
\end{itemize} then $\int_0^1H(t)\dif t$ is the maximum of $\Pr[Y=1]$, and this maximum is attained by $Y$.

Let $\delta_t$ denote the point mass at $t$. We omit the calculation and directly provide the solution $Y$ and its certificate $H$ for each $\ell \in \{2, 3, 4\}$.

\begin{itemize}
    \item $\ell = 2$: $Y_2 = \frac12\delta_0 + \frac12\delta_1$, certified by $H_2(t) = t$. This gives a lower bound of $3$ as $m\to\infty$.
    \item $\ell = 3$: $Y_3 = \frac34\delta_{1/3} + \frac14\delta_1$, certified by $H_3(t) = \frac14(3t-1)^2$. This gives a lower bound of $7/3$ as $m\to\infty$.
    \item $\ell = 4$: $Y_4 = \frac16\delta_0 + \frac23\delta_{1/2} + \frac16\delta_1$, certified by $H_4(t) = t(2t - 1)^2$. This gives a lower bound of $11/5$ as $m\to\infty$.\qedhere
\end{itemize}
\end{proof}

We end this section with an omitted proof.

\begin{proof}[Proof of \Cref{lem:metric-distortion-ordinary}]
By \Cref{def:biased}, if a voter $v$ ranks $a^\star$ first, then $d(a^\star, v) = 0$; otherwise, some ordinary candidate is ranked above $a^\star$, so $d(a^\star, v) = 1/2$. Moreover, for any ordinary candidate $b$, we have $d(b, v) - d(a^\star, v) = \min_{c \preceq_v b} \rho(c) = \mathbf{1}[a^\star \succ_v b]$. Since $a^\star$ is ranked first with probability $\Ev[Y^m]$, the cost of the optimal candidate $a^\star$ is
\[
\SC(a^\star) = \Ev[Y^m] \cdot 0 + (1 - \Ev[Y^m]) \cdot \frac12 = \frac{1 - \Ev[Y^m]}2.
\] To compute $\SC(b) - \SC(a^\star)$ for an ordinary candidate $b$, we divide the voters into three disjoint groups:
\begin{itemize}
    \item With probability $\Ev[Y^m]$, voter $v$ ranks $a^\star$ first, and then $d(b, v) - d(a^\star, v)= 1$.
    \item With probability $1 - \Ev[Y]$, voter $v$ ranks $b$ above $a^\star$, and then $d(b, v) - d(a^\star, v)= 0$.
    \item With the remaining probability $\Ev[Y] - \Ev[Y^m]$, voter $v$ ranks $a^\star$ above $b$ but not first, and then $d(b, v) - d(a^\star, v)= 1$.
\end{itemize} Therefore, using $\Ev[Y] = 1/2$, we have
\[
\SC(b) - \SC(a^\star) = \Ev[Y^m] + \left(\Ev[Y] - \Ev[Y^m]\right) = \Ev[Y] = \frac12,
\]
and hence
\[
\frac{\SC(b)}{\SC(a^\star)} = 1 + \frac{1/2}{(1 - \Ev[Y^m])/2} = 1 + \frac{1}{1 - \Ev[Y^m]}.\qedhere
\]
\end{proof}

\section{Proofs and scripts for Theorem \ref{thm:optimal-g}}
\label{app:optimal-g}

This appendix completes the proof of \Cref{thm:optimal-g}. In \Cref{sec:boundary}, we prove the inequalities of the Verifier Lemma by hand in the cases that the computer check behind \Cref{lem:certificate} does not cover. In \Cref{app:verification}, we describe the two scripts cited in \Cref{prop:lower-certificate,lem:certificate}.

\subsection{Masses near 0 and 1}
\label{sec:boundary}
 
Recall from Step 4 in \Cref{sec:upper-proof} that \Cref{eq:low-cuts} collects the conditions for the voters with $s_v<\lambda$, and \Cref{eq:high-cuts} those for the voters with $s_v\geq\lambda$, at each value of the mass $p=p_t=D(S_t)$. We check these conditions with \href{https://github.com/YQL-Skorpion/stable-lotteries/blob/main/section_6_checker/verify.py}{\texttt{verify.py}} for masses $p$ in $[0.01,0.99]$ (\Cref{app:verification}). Unfortunately we cannot check all of $[0,1]$ this way: at $p=0$, some of the inequalities in \Cref{eq:low-cuts,eq:high-cuts} degenerate to $0\leq0$, and as $p\to1$, the first inequality in \Cref{eq:low-cuts} becomes tight along an interval of $q$. Hence, in this section, for $p$ close to $0$ or $1$, we instead choose $\beta$ and $c$ in closed form and prove the conditions by hand.

We introduce two pieces of notation for convenience. First, let $\sigma(x)=\lambda-x-\phi(x)\geq0$ be the gap between $\phi$ and the line $\lambda-x$ from Election 1. It vanishes exactly where $\phi$ runs along this line, as in region 1 of \Cref{fig:certificate}. Second, for $x<\lambda$, let
\[
r_x=\frac{\lambda-x}{\phi(x)}\geq1,
\qquad
u_x=\frac{c(p,x)}{\phi(x)}\geq0 .
\]
The line through $(x,\phi(x))$ and $(\lambda,0)$ in \Cref{fig:potential-geometry} has slope $-1/r_x$, and \Cref{eq:potential} reads $W_x(y)=0$ for $y\leq\lambda$ and $W_x(y)=\lambda-y-r_x\phi(y)$ for $y>\lambda$. Dividing \Cref{eq:high-cuts} by $\phi(x)>0$, the conditions for the voters with $s_v\geq\lambda$ become
\begin{equation}\label{eq:high-cuts-div}
	\begin{aligned}
		q+\beta(p)\delta_p(q)+W_x(q)-\lambda&\leq(-\phi(q))_+u_x
		&&\text{for }0\leq q\leq p,\\
		q+\beta(p)K_p(q,z)+W_x(q)-W_x(z)&\leq(-\phi(z))_+u_x
		&&\text{for }0\leq q\leq p\leq z\leq1.
	\end{aligned}
\end{equation}

We also use a few elementary consequences of \Cref{eq:score-shape}. Since $\phi$ is nonincreasing and $\int_0^1\phi=0$, the function $J$ is concave with $J(0)=J(1)=0$, so $J\geq0$. Since $\phi\geq\phi(1)=-1$, we also have $J(y)=\int_y^1(-\phi(r))\,dr\leq1-y$. These facts give the following bounds on $\delta_p(q)$ and $J$:
\begin{equation}\label{eq:elementary-bounds}
	\delta_p(q)\leq p-q,\qquad
	\delta_p(q)\leq(1-p)\,\phi(q),\qquad
	J(q)\geq q\,\phi(q).
\end{equation}
Similarly, for $z\geq p$, since $\phi(r)\geq\phi(z)$ on $[z-p+q,z]$, we have
\begin{equation}\label{eq:K-bound}
	K_p(q,z)\leq-(p-q)\,\phi(z)-J(q).
\end{equation}

For small masses, we take $\beta$ constant. For large masses, we take $\beta(p)=p/J(p)$. With this choice, dividing \Cref{eq:binary} by $J(p)$ gives the first inequality in \Cref{eq:low-cuts} with $c=0$. In both regimes, we take $c$ as small as the conditions for the voters with $s_v\geq\lambda$ allow, since $c$ enters the conditions for the voters with $s_v<\lambda$ with a positive sign.

\begin{lemma}\label{lem:boundary}
	Let $\phi$ satisfy the hypotheses of \Cref{lem:verifier}, and let $0<p_0<\lambda<p_1<1$.
	\begin{enumerate}[label=(\roman*)]
		\item\label{it:small} Let $\beta_0>0$ satisfy
		\begin{equation}\label{eq:small-hyp}
			\begin{gathered}
				\beta_0\phi(p_0)\geq1,\qquad
				p_0\left(1+\beta_0(1+\phi(0))\right)\leq\lambda,\\
				p_0\left[1-\beta_0\left(\phi(p_0)-\phi(x)\right)\right]\leq\sigma(x)\qquad\text{for }0\leq x\leq\lambda.
			\end{gathered}
		\end{equation}
		Then \Cref{eq:low-cuts,eq:high-cuts} hold for every $p\in(0,p_0]$ with $\beta(p)=\beta_0$ and $c(p,x)=\left((1+\beta_0p)\phi(x)-(\lambda-x)\right)_+$.
		\item\label{it:large} Suppose that $\phi(r)=\lambda-r$ on an interval $[r_-,z_+]$ with $r_-<\lambda$ and $\lambda+(1-p_1)\leq z_+$, that $\phi$ is concave on $[\lambda,1]$, and that for constants $0<m_0<1\leq M$,
		\begin{equation}\label{eq:lipschitz}
			\begin{aligned}
				\phi(r)-\phi(s)&\leq m_0(s-r)&&\text{for }0\leq r<s\leq r_-,\\
				\phi(r)-\phi(s)&\leq M(s-r)&&\text{for }p_1\leq r<s\leq1.
			\end{aligned}
		\end{equation}
		If, for every $p\in[p_1,1)$, \Cref{eq:binary} holds and
		\begin{equation}\label{eq:large-hyp}
			m_0\,p(1-p)\leq J(p)\leq p(1-p),\qquad
			p\left(1-(M+1)(1-p)\right)\geq(M-1)J(p),
		\end{equation}
		then \Cref{eq:low-cuts,eq:high-cuts} hold for every $p\in[p_1,1)$ with $\beta(p)=p/J(p)$ and $c(p,x)=\left(\lambda-x-\beta(p)(1-p)\phi(x)\right)_+$.
	\end{enumerate}
\end{lemma}

\Cref{lem:boundary} reduces the masses near $0$ and $1$ to conditions on $\phi$ alone. For $\phi^\star$, with $p_0=0.01$ and $p_1=0.99$, we check these conditions with \href{https://github.com/YQL-Skorpion/stable-lotteries/blob/main/section_6_checker/verify.py}{\texttt{verify.py}} (\Cref{app:verification}).

\begin{proof}[Proof of \ref{it:small}]
	Fix $p\in(0,p_0]$, $0\leq q\leq p$, and $x\in[0,\lambda)$, and write $\eta=1+\beta_0p$, so that $c(p,x)=(\eta\phi(x)-(\lambda-x))_+$. We will also use that
	\[
	q-\beta_0J(q)\leq q\left(1-\beta_0\phi(q)\right)\leq0,
	\]
	which follows from $J(q)\geq q\phi(q)$ in \Cref{eq:elementary-bounds}, $\phi(q)\geq\phi(p_0)$, and the first condition in \Cref{eq:small-hyp}.

	\begin{itemize}
	\item \emph{First inequality in \Cref{eq:low-cuts}.} Since $\sigma(x)\geq0$, we have $c(p,x)=\left(\beta_0p\,\phi(x)-\sigma(x)\right)_+\leq\beta_0p\,\phi(0)$. Together with $\delta_p(q)\leq p-q$, this gives
	\[
	q+\beta_0\delta_p(q)+c(p,x)\leq p\left(1+\beta_0(1+\phi(0))\right)\leq\lambda
	\]
	by the second condition in \Cref{eq:small-hyp}.

	\item \emph{Second inequality in \Cref{eq:low-cuts}.} Here $p\leq x<\lambda$, so $\phi(x)\geq0$, and \Cref{eq:K-bound} gives $q+\beta_0K_p(q,x)\leq q-\beta_0J(q)\leq0$. It remains to treat the case $c(p,x)=\eta\phi(x)-(\lambda-x)$. Applying \Cref{eq:K-bound} again,
	\begin{align*}
	q+\beta_0K_p(q,x)+\eta\phi(x)-(\lambda-x)
	&\leq q-\beta_0J(q)+(1+\beta_0q)\phi(x)-(\lambda-x)\\
	&\leq q\left[1-\beta_0\left(\phi(q)-\phi(x)\right)\right]-\sigma(x).
	\end{align*}
	If $1-\beta_0\left(\phi(q)-\phi(x)\right)\leq0$, the right-hand side is at most $-\sigma(x)\leq0$. Otherwise, since $\phi(q)\geq\phi(p_0)$ and $q\leq p_0$, the right-hand side is at most $p_0[1-\beta_0(\phi(p_0)-\phi(x))]-\sigma(x)$, which is nonpositive by the last condition in \Cref{eq:small-hyp}.

	\item \emph{First inequality in \Cref{eq:high-cuts-div}.} Since $q\leq p\leq p_0<\lambda$, we have $W_x(q)=0$ and $(-\phi(q))_+=0$, so the inequality reduces to $q+\beta_0\delta_p(q)\leq\lambda$, which follows from the bound above, since $c(p,x)\geq0$.

	\item \emph{Second inequality in \Cref{eq:high-cuts-div}.} Again $W_x(q)=0$. If $z\leq\lambda$, then $W_x(z)=0=(-\phi(z))_+$, and the inequality reduces to $q+\beta_0K_p(q,z)\leq0$, which holds as in the second inequality in \Cref{eq:low-cuts}, since only $\phi(z)\geq0$ was used there. If $z>\lambda$, substituting $W_x(z)=\lambda-z-r_x\phi(z)$ rewrites the inequality as
	\[
	q+z-\lambda+\beta_0K_p(q,z)+(r_x+u_x)\phi(z)\leq0.
	\]
	Since $u_x=(\eta-r_x)_+$, we have $r_x+u_x\geq\eta$. As $\phi(z)\leq0$, it suffices to show that $q+z-\lambda+\beta_0K_p(q,z)+\eta\phi(z)\leq0$. By \Cref{eq:K-bound}, the left-hand side is at most
	\[
	q\left(1-\beta_0\phi(q)\right)+\left(z-\lambda+\phi(z)\right)+\beta_0q\,\phi(z),
	\]
	and each of these three terms is nonpositive: the first by the bound $q(1-\beta_0\phi(q))\leq0$ above, the second by $\phi(z)\leq\lambda-z$, and the third by $\phi(z)\leq0$.
	\end{itemize}
\end{proof}

\begin{proof}[Proof of \ref{it:large}]
	Fix $p\in[p_1,1)$ and write $h=1-p$, $\beta=p/J(p)$, and $t=\beta h$, so that $c(p,x)=(\lambda-x-t\phi(x))_+$ and $u_x=(r_x-t)_+$. The two bounds on $J(p)$ in \Cref{eq:large-hyp} give $1\leq t\leq1/m_0$. Dividing \Cref{eq:binary} at $p$ by $J(p)>0$ gives
	\begin{equation}\label{eq:binary-at-p}
		q+\beta\delta_p(q)\leq\lambda\qquad\text{for }0\leq q\leq\lambda.
	\end{equation}

	\begin{itemize}
	\item \emph{First inequality in \Cref{eq:low-cuts}.} Here $q\leq x<\lambda$. Since $c(p,x)=\max\{0,\lambda-x-t\phi(x)\}$, besides \Cref{eq:binary-at-p} we need
	\[
	q+\beta\delta_p(q)-\lambda+\lambda-x-t\phi(x)=(q-x)+t\left[\frac{\delta_p(q)}{h}-\phi(x)\right]\leq0 .
	\]
	If $x\geq r_-$, then $\phi(x)=\lambda-x$, so $\lambda-x-t\phi(x)=(1-t)(\lambda-x)\leq0$ and the left-hand side is at most $q+\beta\delta_p(q)-\lambda\leq0$ by \Cref{eq:binary-at-p}. If $x<r_-$, then $\delta_p(q)/h\leq\phi(q)$ by \Cref{eq:elementary-bounds} and $\phi(q)-\phi(x)\leq m_0(x-q)$ by \Cref{eq:lipschitz}, so the left-hand side is at most $(x-q)(tm_0-1)\leq0$, because $t\leq1/m_0$.

	\item \emph{Second inequality in \Cref{eq:low-cuts}.} This inequality only applies when $p\leq x<\lambda$, so it is vacuous here, as $p\geq p_1>\lambda$.

	\item \emph{First inequality in \Cref{eq:high-cuts-div}.} If $q\leq\lambda$, then $W_x(q)=0=(-\phi(q))_+$, and the inequality is \Cref{eq:binary-at-p}. If $q>\lambda$, substituting $W_x(q)=\lambda-q-r_x\phi(q)$ turns the inequality into $\beta\delta_p(q)\leq(r_x-u_x)\phi(q)$. Since $u_x=(r_x-t)_+$, we have $r_x-u_x=\min\{r_x,t\}$. As $\phi(q)<0$ and $\delta_p(q)\leq h\phi(q)$ by \Cref{eq:elementary-bounds}, we get
	\[
	\beta\delta_p(q)\leq\beta h\,\phi(q)=t\,\phi(q)\leq\min\{r_x,t\}\,\phi(q).
	\]

	\item \emph{Second inequality in \Cref{eq:high-cuts-div}.} Here $z\geq p>\lambda$, so $(-\phi(z))_+=-\phi(z)$ and $W_x(z)=\lambda-z-r_x\phi(z)$.
		\begin{itemize}
		\item \emph{Case 1: $q\leq\lambda$.} Then $W_x(q)=0$, and the inequality reads $q+z-\lambda+\beta K_p(q,z)+(r_x+u_x)\phi(z)\leq0$. Since $r_x+u_x\geq r_x\geq1$ and $\phi(z)\leq0$, it suffices to show that
	\[
	E(z)=q+z-\lambda+\beta K_p(q,z)+\phi(z)\leq0\qquad\text{for }p\leq z\leq1.
	\]
	At $z=1$, we have $K_p(q,1)=\delta_p(q)$ and $\phi(1)=-1$, so $E(1)=q+\beta\delta_p(q)-\lambda\leq0$ by \Cref{eq:binary-at-p}. It remains to show that $E$ is nondecreasing on $[p,1]$. For $p\leq z<z'\leq1$, since $\phi$ is nonincreasing and satisfies \Cref{eq:lipschitz},
	\begin{align*}
	E(z')-E(z)&=(z'-z)+\beta\left[\int_{z-p+q}^{z'-p+q}\phi(u)\,du-\int_z^{z'}\phi(u)\,du\right]+\phi(z')-\phi(z)\\
	&\geq(z'-z)\left[1+\beta\left(\phi(z'-p+q)-\phi(z)\right)-M\right].
	\end{align*}
	Since $z'-p+q\leq\lambda+h\leq z_+$, we have $\phi(z'-p+q)\geq\phi(\lambda+h)=-h$. Also, $\phi(z)\leq\phi(p)\leq\phi(1)+Mh=-1+Mh$ by \Cref{eq:lipschitz}. Hence $E(z')-E(z)\geq(z'-z)\left[1-M+\beta(1-(M+1)h)\right]\geq0$ by the second condition in \Cref{eq:large-hyp}, and so $E(z)\leq E(1)\leq0$ for all $z\in[p,1]$.

		\item \emph{Case 2: $q>\lambda$.} Let $s=z-p\geq0$. Substituting $W_x(q)$ and $W_x(z)$, the inequality becomes $z+\beta K_p(q,z)-r_x[\phi(q)-\phi(z)]+u_x\phi(z)\leq0$. Since $u_x\phi(z)\leq0$, it suffices to show that $G(q)\leq0$, where
	\begin{align*}
	G(q)&=z+\beta K_p(q,z)-r_x\left[\phi(q)-\phi(z)\right],\\
	K_p(q,z)&=\int_0^s\left[\phi(q+u)-\phi(p+u)\right]\,du-J(p).
	\end{align*}
	\begin{itemize}
		\item If $\beta s\leq r_x$, concavity of $\phi$ on $[\lambda,1]$ and $q\leq p$ imply that $\phi(q+u)-\phi(p+u)$ is nondecreasing in $u$, so the integral in $K_p(q,z)$ is at most $s[\phi(q+s)-\phi(z)]$. Using $\beta J(p)=p$, then $\beta s\leq r_x$ together with $\phi(q+s)\geq\phi(z)$, and finally $\phi(q+s)-\phi(q)\leq-s$ (since $\phi(r)+r$ is nonincreasing on $[\lambda,1]$) together with $r_x\geq1$, we obtain
		\begin{align*}
		G(q)&\leq z-p+\beta s\left[\phi(q+s)-\phi(z)\right]-r_x\left[\phi(q)-\phi(z)\right]\\
		&\leq z-p+r_x\left[\phi(q+s)-\phi(q)\right]\leq z-p-s=0.
		\end{align*}
		\item If $\beta s\geq r_x$, let $\phi'$ denote the right derivative of $\phi$ on $[\lambda,1]$, so that $\phi(q+s)-\phi(q)\leq s\phi'(q)$ and $\phi'(q)\leq0$. Then
		\[
		\partial_qG(q)=\beta\left[\phi(q+s)-\phi(q)\right]-r_x\phi'(q)\leq(\beta s-r_x)\,\phi'(q)\leq0,
		\]
		so $G$ is nonincreasing on $[\lambda,p]$. Since $r_x\geq1$ and $\phi(z)\leq0$, it follows that $G(q)\leq G(\lambda)=z+\beta K_p(\lambda,z)+r_x\phi(z)\leq z+\beta K_p(\lambda,z)+\phi(z)$, which is $E(z)$ at $q=\lambda$ and hence nonpositive by Case 1.
	\end{itemize}
		\end{itemize}
	\end{itemize}
\end{proof}

\subsection{The scripts}
\label{app:verification}

This section describes the two data files and the two verification scripts cited in \Cref{sec:optimal-g}.

\paragraph{\href{https://github.com/YQL-Skorpion/stable-lotteries/blob/main/section_6_checker/verify.py}{\texttt{verify.py}}.}
With this script, we check the hypotheses of \Cref{lem:verifier} for $\phi^\star$, $\beta$, and $c$, except for the masses near $0$ and $1$, where we check the hypotheses of \Cref{lem:boundary} instead. The file \href{https://github.com/YQL-Skorpion/stable-lotteries/blob/main/section_6_checker/certificate.json}{\texttt{certificate.json}} stores $\lambda$, the polynomial pieces of $\phi^\star$, the constants in \Cref{lem:boundary}, and the parameters defining $\beta$ and $c$ on $[p_0,p_1]$. There, $c$ has the form
\[
c(p,x)=\left(a(p)(\lambda-x)+b(p)\,\phi(x)\right)_+,\qquad a(p)\in\{-1,+1\},
\]
where $\beta(p)$ and $b(p)$ are given on three subintervals of $[p_0,p_1]$: by the formulas of \Cref{lem:boundary}\ref{it:small}, as piecewise-linear functions, and by the formulas of \Cref{lem:boundary}\ref{it:large} with $\beta$ multiplied by a constant slightly larger than $1$. For $p\in[p_0,p_1]$, we check \Cref{eq:low-cuts,eq:high-cuts-div} on small regions that together cover their domains.

Every hypothesis on $\phi^\star$ in \Cref{lem:verifier,lem:boundary}, except \Cref{eq:binary} for $p\in[p_1,1)$, reduces to finitely many polynomial equalities and inequalities in one variable, which we check directly. (The conditions in \Cref{eq:large-hyp} are polynomial in $p$ on $[p_1,1]$ because $J^\star$ is a single polynomial there.)

The remaining hypothesis, that \Cref{eq:binary} holds for all $p\in[p_1,1)$ and $q\in[0,\lambda]$, involves two variables, and we reduce it to polynomial inequalities in $p$. Write $h=1-p$ and $h_0=1-p_1\geq h$. Then \Cref{eq:binary} at $(p,q)$ reads $p\delta_p(q)\leq(\lambda-q)J(p)$, where $\delta_p(q)=h(\lambda-q)-h^2/2-\int_q^{q+h}\sigma(u)\,du$. We first check that $\phi^\star(r)-\phi^\star(s)\leq s-r$ for $r<s\leq z_+$, so that $\sigma$ is nonincreasing on $[0,z_+]\supseteq[q,q+h]$. Fixing a partition $0=q_0<\dots<q_K=\lambda$ from the file, let $s_i=\sigma(\min\{q_{i+1}+h_0,\lambda\})$. For $q\in[q_i,q_{i+1}]$, we have $q+h\leq q_{i+1}+h_0$, so $\int_q^{q+h}\sigma(u)\,du\geq h\,s_i$, and therefore
\[
p\,\delta_p(q)-(\lambda-q)J(p)\leq(\lambda-q)\left(ph-J(p)\right)-\frac12ph^2-ph\,s_i .
\]
The right-hand side is affine in $q$, so it suffices to check that it is nonpositive at $q=q_i$ and $q=q_{i+1}$. At each of these points, it is a polynomial in $p$ on $[p_1,1]$, which leaves $2K$ polynomial inequalities.

\paragraph{\href{https://github.com/YQL-Skorpion/stable-lotteries/blob/main/section_6_checker/verify_lower_weights.py}{\texttt{verify\_lower\_weights.py}}.}
With this script, we check that the linear system $S(\lambda,n)$ from \Cref{sec:lower-g} is infeasible. The file \href{https://github.com/YQL-Skorpion/stable-lotteries/blob/main/section_6_checker/lower_certificate.json}{\texttt{lower\_certificate.json}} records $\lambda=0.56162$, the grid size $n$, and weights for the conditions of $S(\lambda,n)$, with nonnegative weights on the inequalities. Adding up the conditions with these weights gives a linear inequality $w\cdot(\phi,J)\geq w_0$ that holds for every solution $(\phi,J)$ of $S(\lambda,n)$. We then check that $w\cdot(\phi,J)<w_0$ whenever $(\phi,J)$ satisfies \hyperref[it:S-bounds]{$S(\lambda,n)$'s condition (v)}, which is a contradiction.

\end{document}